\documentclass{article}
\usepackage{amsfonts}
\usepackage{amssymb}
\usepackage{amsmath}
\usepackage[dvips,left=1.2in, right=1.2in, top=1.2in, bottom=1in]{geometry}

\newtheorem{theorem}{Theorem}

\newtheorem{definition}[theorem]{Definition}

\newtheorem{proposition}[theorem]{Proposition}
\newtheorem{remark}[theorem]{Remark}

\newenvironment{proof}[1][Proof]{\noindent \textbf{#1.} }{\  \rule{0.5em}{0.5em}}
\begin{document}

\title{Competition in fund management and forward relative performance criteria\footnote{A considerable part of this work is part of the Ph.D.~Thesis (2017) of the second author (\cite{Gen17}) and first appeared in https://ssrn.com/abstract=2870040. The work has been presented at conferences and seminars (Oxford, Brown University and National and Technical University of Athens). The authors would like to thank the participants for their
valuable comments and suggestions.}}

\author{Michail Anthropelos\thanks{Department of Banking and Financial Management, University of Piraeus; anthropel@unipi.gr.}, \ Tianran Geng\thanks{Boston Consulting Group (BCG), Southborough, MA;
iyagtr@gmail.com.} \ and Thaleia Zariphopoulou\thanks{Departments of Mathematics and
IROM, The University of Texas at Austin and the Oxford-Man Institute of
Quantitative Finance; zariphop@math.utexas.edu.}}
\maketitle

\begin{abstract}
In an It\^{o}-diffusion market, two fund managers trade under relative performance concerns. For both the asset specialization and diversification settings, we analyze the passive and competitive cases. We measure the performance of the managers' strategies via forward relative performance criteria, leading to the respective notions of forward best-response criterion and forward Nash equilibrium. The motivation to develop such criteria comes from the need to relax various crucial, but quite stringent, existing assumptions - such as, the a priori choices of both the market model and the investment horizon, the commonality of the latter for both managers as well as the full a priori knowledge of the competitor's policies for the best-response case. We focus on locally riskless criteria and deduce the random forward equations. We solve the CRRA cases, thus also extending the related results in the classical setting. An important by-product of the work herein is the development of forward performance criteria for investment problems in It\^{o}-diffusion markets under the presence of correlated random endowment process for both the perfectly and the incomplete market cases.
\end{abstract}

\newpage

\section{Introduction}

Relative performance is of tantamount importance in both the mutual and
hedge fund management industries. It impacts a variety of factors spanning
from a company's reputation and net cash inflows to the incentive
structure and promotion schedule for its managers. Such facts have been very well
documented in the finance practice and have been extensively studied in
academic research (see, among others, \cite{ChePen09, Chevalier99, Gallaher06, KemRueTan09, ST98}).

While classifying the various kinds patterns of relative performance among
managers is rather complex (classification by sectors, asset riskiness,
market conditions, business cycles and others), there is a prevailing
dichotomy based on \textit{asset specialization} or \textit{asset
diversification}. In the former, the competing managers specialize in
distinct asset classes while, in the latter, they invest in common ones.

Asset specialization stems from a variety of reasons like familiarity with a
certain sector, reduction of costs to enhance knowledge of new stocks,
trading costs and constraints (\cite{Bre75}), liquidity costs (\cite{Wagner11}), and ambiguity aversion (\cite{CaoHenWan05, MukTall01}). The above evidence has been well established
in the empirical literature (\cite{CM99, FerKesMigRam13, K05}) and has been incorporated in several theoretical models (\cite{AheWilWar04, FrePot91, Merton87, VanNieVeld09}; see also \cite{BH08,CM99, K05,
Merton87, MV07, UW03, vanB08, Wagner11}).

In asset diversification, the motivation is mainly to increase the net money inflows from clients (\cite{Basak07, CK11}). This setting is also more suitable to
model relative performance concerns against a given benchmark portfolio
(typically, a mix of asset classes). In a different direction, asset
diversification also occurs in delegated portfolio management, where the
role of one manager is replaced by the client (\cite{Kaniel17}). In a
related family of models, it also appears in the so-called \textquotedblleft
catching up with the Joneses\textquotedblright \ literature (see \cite%
{Abel90, Gomez09}).

Relative performance has been also considered in terms of how each manager reacts to the performance of a competitor. This interaction can be \textit{%
passive (best-response)} in that the manager takes the competitors' policies as given (arbitrary but fixed) and trades accordingly, without any further interaction (see, among others, \cite{ChePen09, KemRueTan09, KoskPon99}). On the other hand, interaction may be also \textit{competitive}, when managers compete
with each other dynamically while investing among the various accounts (\cite{Chevalier99, Gallaher06}).

Whether managers compete within the same or different asset classes and/or
interact in a passive or competitive manner, there are common underlying
assumptions that limit the generality and applicability of the existing
studies. The aim herein is to revisit some of these assumptions, propose an
alternative approach, study the related optimization problems and develop a
comparative study with previous works. We are motivated to do so not only by
theoretical and conceptual arguments but, also, by various recent empirical
works that point out to strong dependencies of the observed policies to
dynamically evolving factors, a dependency that cannot be explained in
traditional settings; see, for example, \cite{HuaSiaCle11, KemRueTan09}, where the effects of the current (and possibly non-anticipated) phase of the market on the managers' behavior is discussed.

The first such assumption, ubiquitous for solving all underlying expected
utility problems, is that the market model is a priori chosen for the entire
duration of the investment activity (see, for example, \cite{Basak15, BieArnReis17}). This
is, however, rather unrealistic since model error and model decay always
occur. Of course, genuinely dynamic model revision may be incorporated (for
example, in the context of adaptive control), but
then, intertemporal consistency is violated. We also note that even in the
popular robust case, widely used to remedy model uncertainty and ambiguity,
there is a stringent underlying assumption that the plausible family of
models is itself a priori chosen. Similar restrictions are also present in
filtering, in that the associated observation process is also pre-chosen.

The second assumption is related to the investment horizon choice. In all
existing works, it is assumed that the horizon is i)\ a single one (finite
or infinite), ii)\ a priori chosen and iii) common across competitors (see, for example, \cite{Basak15, BieArnReis17, LacZar19}). In practice, however, this is not
the case. While it is customary for managers to report their performance at
common standardized time intervals (e.g. quarterly, annually), they almost
always have internal sub-horizons that depend on company-related factors and
which are themselves difficult to model. Furthermore, even if a common
horizon is a priori chosen, the investment activity does not stop at this
specific pre-assigned time, as managers always roll their positions from one
investment horizon to the next. One could then argue that managers apply the
same (or very similar) goals for the upcoming period. This, however, is 
\textit{not} supported by existing empirical evidence which shows that
managers always adapt their goals in a rather complex manner, depending on
realized losses and gains, new upcoming (frequently unpredictable) market
conditions and others (see, for example, \cite{Abel90, BarHuaOde16, BerBins, BodSim16, DonFenSad19}).

The third assumption is related to managers' interaction. It is always
assumed that each manager has full, and for the entire investment horizon,
knowledge of the competitor's policy, in that she knows the stochastic
process that models the investment policy she competes against. This
modeling input is needed in order to solve all associated expected utility
problems but as an assumption, it is quite unrealistic. Indeed, not only the
manager cannot have such foreknowing skills for her competitor, the
competitor herself might not a priori know how her own strategy will be
changing as the market enfolds.

Herein, we propose a new framework aiming to remedy some of the shortcomings
of these three stringent assumptions. We make no specific assumptions about
the market model besides a weak structural one that the asset prices are
It\^o-diffusion processes whose coefficients adapt to the current information
(see (\ref{stocks})). We also make no specific assumption about any
pre-specified investment horizon, allowing for each manager to invest till
personalized discretionary times. Finally, we make no assumption on a priori
choosing the stochastic process that yields the competitor's policy. Rather,
we allow (besides mild integrability conditions) this policy to be
dynamically revealed to her competitor. For tractability, we only consider
the case with two managers. The general case for $N$-managers as well as the
mean field game limit are left for future research\footnote{One may consider the case of infinite competitors and build the notion of a ``forward'' mean-field game (MFG). However, constructing such a notion is not immediate as various concepts might not ``carry over'', especially when there is common noise and/or the forward MFG performance process has its own volatility in a general It\^{o}-diffusion setting. To date, the proper definition of a forward MFG has not been produced, and neither the convergence of the forward finite game to the forward MFG. Formally, one may mimic the definitions herein (see also \cite{Gen17}) and the ones in \cite{LacZar19} for the classical case, and calculate a special solution within the CARA functions (see \cite{ReiPla20}).}.

The new framework is built on extensions of the so-called \textit{forward}
performance criteria. Such criteria, introduced by one of the authors and
M.~Musiela, and further developed by others (see \cite{Musiela-first, MZ-QF} and \cite{Zitkovic}), are
modeled as stochastic processes, say $\left( U\left( x,t\right) \right)
_{t\geq 0},$ that adapt to the market information, are (local)
supermartingales along all admissible policies and (local) martingales along
an optimal policy. To characterize forward criteria in such markets, a
stochastic PDE was proposed in \cite{MZ10a}. Depending on the choice of
its volatility and structural parametrizations, various forms of $U(x,t)$
have been studied (see, among others, \cite{AvaShkSir18, Liang, Shkolnikov}). However,
several questions remain open as the underlying stochastic optimization
problems are ill-posed, fully non-linear and degenerate.

To build forward criteria that allow for interaction - passive or
competitive - between two decision makers, we proceed as follows. Let us
assume that each manager uses admissible policies $\alpha $ and $\beta ,$
generating wealths $X_{1}^{\alpha }$ and $X_{2}^{\beta }$.

For the case of best-response, we introduce the \textit{best-response} 
\textit{forward }criterion for manager $1$ as a process $%
U_{1}(x_{1},x_{2},t;\beta )$ such that $U_{1}(X_{1}^{\alpha },X_{2}^{\beta
},t;\beta )$ is a (local) supermartingale for each admissible policy $\alpha 
$ and becomes a (local) martingale, $U_{1}(X_{1}^{\alpha ^{\ast
}},X_{2}^{\beta },t;\beta ),$ along an optimal $\alpha ^{\ast }$. We stress
that in contrast to \textit{all }classical cases, the competitor's policy
process $\beta $ is \textit{not} pre-assumed. Rather it is being revealed in ``real time" and, in turn, the performance criterion $U_{1}(x_{1},x_{2},t;%
\beta )$ adapts to it dynamically. The best-response forward criterion for
manager $2,$ $U_{2}(x_{1},x_{2},t;\alpha ),$ is defined analogously, and
with the competitor's policy process $\alpha $ also not a priori known.

For the case of competitive interaction between the managers, we introduce a 
\textit{forward Nash equilibrium} criterion, consisting of two pairs\textit{%
\ }$\left( U_{1}(x_{1},x_{2},t;\beta )_{t\geq 0},\left( \alpha _{t}^{\ast
}\right) _{t\geq 0}\right) $ and $\left( U_{2}(x_{1},x_{2},t;\alpha )_{t\geq
0},\left( \beta _{t}^{\ast }\right) _{t\geq 0}\right) $ such that $%
U_{1}(X_{1}^{\alpha ^{\ast }},X_{2}^{\beta },t;\beta )$ and $%
U_{1}(X_{1}^{\alpha },X_{2}^{\beta ^{\ast }},t;\alpha )$ are (local)
supermartingales, and $U_{1}(X_{1}^{\alpha ^{\ast }},X_{2}^{\beta ^{\ast
}},t;\alpha ^{\ast })$ and $U_{2}(X_{1}^{\alpha ^{\ast }},X_{2}^{\beta
^{\ast }},t;\beta ^{\ast })$ are (local) martingales.

For each kind of interaction, based on best response or on competition, we
analyze both the asset specialization and the asset diversification cases.
Herein, we focus on forward criteria that are locally riskless processes,
namely, of the form 
\begin{equation}
dU_{1}(x_{1},x_{2},t;\beta )=b_{1}\left( x_{1},x_{2},t;\beta \right) dt\text{
\ and \ }dU_{2}(x_{1},x_{2},t;\alpha )=b_{2}\left( x_{1},x_{2},t;\alpha
\right) dt\text{,}  \label{U-general}
\end{equation}%
for suitable adapted processes $\left( b_{1}\left( x_{1},x_{2},t;\beta
\right) \right) _{t\geq 0}$ and $\left( b_{2}\left( x_{1},x_{2},t;\alpha
\right) \right) _{t\geq 0}$. We choose this class because, in the absence of
relative concerns, locally riskless forward criteria were the first to be
extensively analyzed not only because of their tractability but, also, for
the valuable intuition in terms of num\'{e}raire choice, time monotonicity of
preferences, dependence on market performance, and others (see \cite{MZ10a} for details).

Throughout, we model the market having one riskless bond and two risky
securities representing proxies of two asset classes. Such proxies have been
consistently used in the literature (see, for example, \cite{Basak15, HeKri13, Koijen14}).
We model their prices as It\^{o}-diffusion processes (cf.~(\ref{stocks})) but we
stress, once more, that their coefficients are not a priori chosen but,
rather, become known gradually, infinitesimally in time, as the market
evolves.

When managers invest in isolation, their wealths evolve as in (\ref{wealths}%
) and in (\ref{wealth-1-div}), (\ref{wealth-2-div}), for the asset
specialization and diversification cases, respectively. Under relative
performance, the competitor's wealth needs to be incorporated. One way to do
this was proposed in \cite{Basak14, Basak15, Basak07, Koijen14}), which
we also adopt herein. Namely, we introduce the \textit{relative wealth }%
processes\textit{\ }$( \tilde{X}_{1,t}) _{t\geq 0}$ and $( 
\tilde{X}_{2,t}) _{t\geq 0},$ with 
\begin{equation*}
\tilde{X}_{1}:=\frac{X_{1}}{X_{2}^{\theta _{1}}}\text{ \  \  \ and \  \  \ }%
\tilde{X}_{2}:=\frac{X_{1}}{X_{2}^{\theta _{2}}},
\end{equation*}%
where the \textit{competition biases} $\theta _{1},\theta _{2}\in \left[ 0,1%
\right] $ model the degree of relative performance considerations.

The limiting case $\theta _{1}=0$ (resp. $\theta _{2}=0)$ expresses that
manager $1$ (resp. $2)$ is not at all concerned with the output of the
opponent. The other limiting case, $\theta _{1}=1$ (resp. $\theta _{2}=1)$
corresponds to the traditional relative performance in terms of a benchmark
(such as S\&P500 index, collection of index funds, and others; see, for example, the
related discussion in \cite[Section 1]{Basak15}).

The form of the relative state dynamics $\tilde{X}_{1}$ and $\tilde{X}_{2}$,
see (\ref{relative-dynamics1}), (\ref{relative-dynamics2}) for the asset
specialization and (\ref{X-1-div}), (\ref{X-2-div})\ for the asset
diversification cases, prompts us to introduce \textit{``personalized"
fictitious} markets and define the relevant forward criteria within.
Informationally, the original and these virtual markets do not differ but
the forward criteria may have different characteristics, depending on the
choice of the modified state variables.

The above choice of $\tilde{X}_{1}$ and $\tilde{X}_{2}$ suggests to develop
criteria of the reduced scaled form, namely, 
\begin{equation}
U_{1}(x_{1},x_{2},t;\beta )=V_{1}\left( \frac{x_{1}}{x_{2}^{\theta _{1}}}%
,t;\beta \right) \text{ \  \ and \ }U_{2}(x_{1},x_{2},t;\alpha )=V_{2}\left( 
\frac{x_{2}}{x_{1}^{\theta _{2}}},t;\alpha \right) ,  \label{U-reduced}
\end{equation}%
$\tilde{x}_{1}=x_{1}/x_{2}^{\theta _{1}}$, $\tilde{x}_{2}=x_{2}/x_{1}^{\theta _{2}},$ for suitable processes $\left( V_{1}\left( \tilde{x}%
_{1},t;\beta \right) \right) _{t\geq 0}$ and $\left( V_{2}\left( \tilde{x}%
_{2},t;\alpha \right) \right) _{t\geq 0}.$ Other forms of relative criteria
may be introduced depending on admissibility domains, type of risk
preferences, etc. (see, for example, additive cases in \cite{BieArnReis17, EspTou15, FreReis11, LacZar19}).

In the asset specialization case, neither manager may invest in the asset
class of the competitor. As a result, the market (be the original or the
fictitious ones) is incomplete. Forward criteria for incomplete markets have
been developed before but only when incompleteness comes exclusively from
imperfectly correlated stochastic factors affecting the stocks' dynamics
(see, among others, \cite{Anthro, LeuSirZar12, Liang, Shkolnikov}). Herein, however, the kind of incompleteness is different, for it is
generated by the specialization constraints. These constraints alter the
relative wealth processes $\tilde{X}_{1}$ and $\tilde{X}_{2}$ in a way that
the related dynamics may be interpreted as either including non-zero
constrained allocations (cf.~(\ref{relative-dynamics1}) and (\ref%
{relative-dynamics2})) or, alternatively, having a stream with imperfectly
correlated return (cf.~(\ref{relative-metric-1}) and (\ref{relative-metric-2}%
)). The former interpretation is closer to the original formulation herein.
The latter has a different scope. It shows how the forward criteria under
asset specialization may be used to define analogous criteria for problems
with (imperfectly correlated) random endowment process (also called stochastic income stream). This is a new class of forward criteria, not been
considered so far.

We analyze both the best-response and the competition cases, and introduce
the corresponding \textit{best-response forward} relative performance
criteria. The definitions extend the original ones in \cite{Musiela-first}. We,
in turn, derive a random PDE (cf.~(\ref{PDE-1})) that the (locally riskless)
criterion is expected to satisfy. Its coefficients are adapted processes,
and depend on both the market dynamics and the competitor's policies. In
general, equation (\ref{PDE-1})\ is not tractable unless for the homothetic
class, which we solve. Nevertheless, its solution is used to derive and
represent the optimal policies in a stochastic feedback form (see (\ref%
{optimal-passive-general})).

When dynamic competition is allowed, this naturally leads to the new concept
of \textit{forward Nash equilibrium,} which we introduce in Definition \ref{definition of forward Nash}. To derive the equilibrium policies, one needs to solve a
system of equations, in general intractable due to interdependent
nonlinearities. The homothetic case is solvable and we provide the relevant
policies. Their form resembles the ones in the log-normal case studied in \cite{Basak15}, but it is now derived, under the new criteria, in the
general It\^{o}-diffusion setting.

In the asset diversification case, both managers invest in a common market.
Their relative performance concerns distort the original wealth processes
(cf. (\ref{X-1-div}) and (\ref{X-2-div})) which, as in the previous case,
leads to two distinct personalized fictitious markets, each depending on the
individual competition parameter and the policy of the opponent. Now,
however, each of these markets is complete as investment is allowed in both
stocks with modified risk premia. Forward criteria may, in turn, be defined
as in the asset specialization case and we focus again on locally riskless
ones. The completeness of the markets enables us to extend the results of \cite{MZ10a} and characterize both the relative performance and Nash
equilibrium criteria, their optimal wealth and investment policies in full
generality. The special case of homothetic criteria is also analyzed.

Conceptually, the analyses of the asset specialization and the asset diversification cases are rather similar, in terms of the associated fictitious markets and the optimality criteria. The fundamental difficulty
is in their (in)completeness which affects the tractability of the problem
and the form of the optimal policies. A key difference is that the locally
riskless forward criteria in the asset diversification case are always
time-decreasing while, in the asset specialization case, they are not. We
further elaborate on this later on.

We conclude the introductory section mentioning that an underlying assumption herein - which is also
widely present in the classical literature - is that the managers have
common information for both the market and the opponent's strategies; we
refer the reader to \cite{Basak14, BerBins} for discussion of supporting
arguments for this assumption. While the access to such information is much
more realistic in our setting (as it occurs in ``real-time"), the fact that both managers share common access to it is, in our view, a rather stringent requirement. As the focus herein is to develop the new, forward framework with relative performance, we also adopt this assumption. We provide ideas
how to relax it and future research in this direction in section 4.

The paper is organized as follows. In section 2 we present the asset
specialization case and analyze the forward best-response and the forward
Nash equilibrium cases. In section 3, we analyze the asset diversification
case while in section 4 we conclude and comment on possible extensions.

\section{Asset specialization and forward competition}\label{sec:spec}

The market consists of one (locally) riskless asset and two risky
securities, representing proxies of two distinct asset classes. The prices
of the risky securities, $\left( S_{1,t}\right) _{t\geq 0}$ and $\left(
S_{2,t}\right) _{t\geq 0}$ are It\^{o}-diffusions solving 
\begin{equation}
\frac{dS_{1}}{S_{1}}=\mu _{1}dt+\sigma _{1}dW_{1}\text{ \ and \ }\frac{dS_{2}%
}{S_{2}}=\mu _{2}dt+\sigma _{2}dW_{2},  \label{stocks}
\end{equation}%
with $S_{1,0},S_{2,0}>0$. The processes $\left( W_{1,t}\right) _{t\geq
0},\left( W_{2,t}\right) _{t\geq 0}$ are standard Brownian motions on a
filtered probability space $(\Omega ,\mathcal{F},(\mathcal{F}_{t})_{t\geq 0},%
\mathbb{P}),$ with correlation coefficient $\rho \in \left( -1,1\right) $
and $\mathcal{F}_{t}$ being the filtration generated by $\left(
W_{1},W_{2}\right) $. The market coefficients $\left( \mu _{i,t}\right)
_{t\geq 0}$,$\left( \sigma _{i,t}\right) _{t\geq 0},$ $i=1,2,$ are $\mathcal{%
F}_{t}$-adapted processes with values in $\mathbb{R}$ and $\mathbb{R}_{+}$,
respectively. The riskless asset is a money market account $\left(
B_{t}\right) _{t\geq 0}$ offering positive interest rate $\left(
r_{t}\right) _{t\geq 0}$, an $\mathcal{F}_{t}$-adapted process.

We denote this \textit{original} market by $\mathcal{M}=\left(
B,S_{1},S_{2}\right) .$ The related market price of risk processes, $\left(
\lambda _{1,t}\right) _{t\geq 0}$ and $\left( \lambda _{2,t}\right) _{t\geq
0},$ are given by 
\begin{equation}
\lambda _{1}=\frac{\mu _{1}-r}{\sigma _{1}}\text{ \  \  \ and \  \ }\lambda
_{2}=\frac{\mu _{2}-r}{\sigma _{2}},  \label{sharpe}
\end{equation}%
and assumed to be bounded processes, $0<c\leq \lambda _{1},\lambda _{2}\leq
C<\infty ,$ $t\geq 0,$ for some (possibly deterministic) constants $c,C$.

In this market environment, we consider two asset managers, indexed by $%
i=1,2.$ They \textit{specialize }in assets $S_{1}$ and $S_{2},$
respectively, in that manager $1$ $($resp. $2)$ trades between the riskless
asset and $S_{1}$ $($resp. $S_{2})$. However, both managers have access to
the common filtration $\left( \mathcal{F}_{t}\right) _{t\geq 0}$ (as for for
example in [11] and [12]).

We denote by $\left( X_{1,t}\right) _{t\geq 0},\left( X_{2,t}\right) _{t\geq
0}$ the wealths of manager $1$ and $2$ and by $\left( \alpha _{t}\right)
_{t\geq 0}$ and $\left( \beta _{t}\right) _{t\geq 0}$ the corresponding
self-financing strategies in assets $S_{1}$ and $S_{2}.$ Then, (\ref{stocks}%
) yields 
\begin{equation}
\frac{dX_{1}}{X_{1}}=\sigma _{1}\alpha \left( \lambda _{1}dt+dW_{1}\right) 
\text{ \  \  \ and \ }\frac{dX_{2}}{X_{2}}=\sigma _{2}\beta \left( \lambda
_{2}dt+dW_{2}\right) ,  \label{wealths}
\end{equation}%
with $X_{i,0}=x_{i}>0,$ $i=1,2$; herein, $X_{1},X_{2},\alpha ,\beta $ are
expressed in discounted (by the riskless asset) units.

The set of admissible policies $\mathcal{A}_{1}$ and $\mathcal{A}_{2},$ of
manager $1$ and $2,$ respectively, are defined for $\left( \pi _{t}\right)
_{t\geq 0}=\left( \alpha _{t}\right) _{t\geq 0},\left( \beta _{t}\right)
_{t\geq 0},$ and $i=1,2,$ 
\begin{equation}
\mathcal{A}_{i}=\left \{ \pi :\pi _{t}\in \mathcal{F}_{t},\text{ }\, \mathbb{E}\left[ \int_{0}^{t}\sigma _{i,s}^{2}\pi _{s}^{2}ds\right] <\infty \text{ \
and }X_{i}>0,\text{ }t>0\right \} .  \label{A}
\end{equation}%
The wealth positivity constraint is in accordance to what is frequently observed in the asset management industry (for instance, mutual funds cannot have negative wealth). The measurability of the individual investment
policies reflects the access by both managers to the common information
generated by $\mathcal{F}_{t}$ (see section 4 for a discussion on this
assumption).

We work with \textit{relative wealth }processes with competition parameters $%
\theta _{1},\theta _{2}\in \left( 0,1\right] $\textit{\ }following the
framework of \cite{Basak14, Basak15, vanB08}). Specifically, if
manager $2$ follows an arbitrary strategy $\beta \in \mathcal{A}_{2}$
generating wealth $X_{2}$, the \textit{relative wealth} of manager $1,$ $%
( \tilde{X}_{1,t}) _{t\geq 0},$ is defined as%
\begin{equation}
\tilde{X}_{1}:=\frac{X_{1}}{X_{2}^{\theta _{1}}},  \label{metric1}
\end{equation}%
with $X_{1}$ and $X_{2}$ solving (\ref{wealths}). Symmetrically, the
relative wealth of manager $2$, $(\tilde{X}_{2,t}) _{t\geq 0},$
given an arbitrary strategy $\alpha \in \mathcal{A}_{1}$ of manager $1$
generating wealth $X_{1},$ is defined as 
\begin{equation}
\tilde{X}_{2}:=\frac{X_{2}}{X_{1}^{\theta _{2}}}.  \label{metric2}
\end{equation}

As discussed in the introduction, we introduce three new modeling elements.
Firstly, while we make the structural model assumption (\ref{stocks}), we do
not pre-choose (at initial time) the processes $\mu _{i}$, $\sigma
_{i},i=1,2.$ Secondly, in a similar manner, we do \textit{not} assume that
the competitors' policies, $\alpha $ and $\beta ,$ are a priori chosen
stochastic processes. Rather, each manager learns the market coefficients
and the opponent's strategy as time enfolds. Thirdly, there is no pre-chosen
investment horizon.

The biased benchmark processes $( X_{1,t}^{\theta
_{2}}) _{t\geq 0}$ and $( X_{2,t}^{\theta _{1}}) _{t\geq 0}$
solve, for $\theta _{1},\theta _{2}\in \left( 0,1\right] ,$ 
\begin{equation*}
\frac{dX_{1}^{\theta _{2}}}{X_{1}^{\theta _{2}}}=\sigma _{1}\theta
_{2}\alpha _{1}\left( \left( \lambda _{1}-\frac{1}{2}\left( 1-\theta
_{1}\right) \sigma _{1}\alpha \right) dt+dW_{1}\right)
\end{equation*}%
and 
\begin{equation*}
\frac{dX_{2}^{\theta _{1}}}{X_{2}^{\theta _{1}}}=\sigma _{2}\theta _{1}\beta
_{2}\left( \left( \lambda _{2}-\frac{1}{2}\left( 1-\theta _{1}\right) \sigma
_{2}\beta \right) dt+dW_{2}\right) .
\end{equation*}%
In turn, the relative wealth $\tilde{X}_{1}$ satisfies

\begin{equation}
\frac{d\tilde{X}_{1}}{\tilde{X}_{1}}=\sigma _{1}\alpha \left( \tilde{\lambda}%
_{1,1}dt+dW_{1}\right) -\sigma _{2}\theta _{1}\beta \left( \text{\ }\tilde{%
\lambda}_{1,2}dt+dW_{2}\right) ,\text{ }  \label{relative-dynamics1}
\end{equation}%
and $\tilde{X}_{1,0}=x_{1}/x_{2}^{\theta _{1}},$ $x_{1},x_{2}>0,$
with the processes $( \tilde{\lambda}_{1,1,t}) _{t\geq 0}$ and $( \tilde{\lambda}_{1,2,t}) _{t\geq 0},$ 
\begin{equation}
\tilde{\lambda}_{1,1}:=\lambda _{1}-\rho \sigma _{2}\theta _{1}\beta \text{
\  \ and \ }\tilde{\lambda}_{1,2}:=\lambda _{2}-\frac{1}{2}\sigma _{2}\left(
1+\theta _{1}\right) \beta .  \label{lamda-1}
\end{equation}%
Symmetrically, the relative wealth $\tilde{X}_{2}$ satisfies 
\begin{equation}
\frac{d\tilde{X}_{2}}{\tilde{X}_{2}}=-\sigma _{1}\theta _{2}\alpha \left( 
\tilde{\lambda}_{2,1}dt+dW_{1}\right) +\sigma _{2}\beta \left( \tilde{\lambda%
}_{2,2}dt+dW_{2}\right) ,\text{ }  \label{relative-dynamics2}
\end{equation}%
and $\tilde{X}_{2,0}=\frac{x_{2}}{x_{1}^{\theta _{2}}},$ $x_{1},x_{2}>0,$
with the processes $( \tilde{\lambda}_{2,1,t}) _{t\geq 0}$ and $%
( \tilde{\lambda}_{2,2,t}) _{t\geq 0}$,%
\begin{equation} 
\text{\ }\tilde{\lambda}_{2,1}:=\lambda _{1}-\frac{1}{2}\sigma _{1}\left(
1+\theta _{2}\right) \alpha \text{ \  \  \ and \  \ }\tilde{\lambda}%
_{2,2}:=\lambda _{2}-\rho \sigma _{1}\theta _{2}\alpha .  \label{lambda-2}
\end{equation}

We may now interpret the relative wealth dynamics (\ref{relative-dynamics1}%
)\ as follows. In the original market $\mathcal{M}=\left(
B,S_{1},S_{2}\right) ,$ manager $1$ chooses (proportional) risky allocations 
$\left( \alpha ,0\right) $ in securities $S_{1}$ and $S_{2},$ due to
specialization. In the relative formulation, it is as if she invests in a
``personalized" \textit{fictitious} market\footnote{%
The superscript ``s" corresponds to specialization.} $\mathcal{\tilde{M}}%
_{1}^{s}:=\left( B,\tilde{S}_{1,1},\tilde{S}_{1,2}\right) $ with (pseudo)
stocks $( \tilde{S}_{1,1,t}) _{t\geq 0},$ $( \tilde{S}%
_{1,2,t}) _{t\geq 0}$ solving (in discounted units) 
\begin{equation}
\frac{d\tilde{S}_{1,1}}{\tilde{S}_{1,1}}=\sigma _{1}\left( \tilde{\lambda}%
_{1,1}dt+dW_{1}\right) \text{ \  \  \ and \  \ }\frac{d\tilde{S}_{1,2}}{\tilde{S%
}_{1,2}}=\sigma _{2}\left( \text{\ }\tilde{\lambda}_{1,2}dt+dW_{2}\right) ,
\label{pseudo-stocks-1}
\end{equation}%
with modified Sharpe ratios $\tilde{\lambda}_{1,1}$ and \ $\tilde{\lambda}%
_{1,2}$ defined in (\ref{lamda-1}). In this virtual market, the original
specialization constraint is not binding, as the manager may now invest in
both risky securities, $\tilde{S}_{1,1}$ and $\tilde{S}_{1,2},$ with
respective proportional allocations $\left( \alpha ,-\theta _{1}\beta
\right) ,$ with only $\alpha $ being controlled by manager $1$. The
constrained allocation $-\theta _{1}\beta $ depends on\textit{\ }both
managers' characteristics, statically on the bias parameter $\theta _{1}$
(chosen by manager $1)$ and dynamically on $\beta $ (chosen by manager $2)$.

Alternatively, we may view (\ref{relative-dynamics1}) as wealth dynamics in
market $\mathcal{\tilde{M}}_{1}^{s}$ where manager $1$ invests in the
riskless security $B$ and chooses ratio $\alpha $ to allocate in the
fictitious stock $\tilde{S}_{1,1},$ while he receives a \textit{process of random endowment} returns $( \tilde{Y}_{1,t}) _{t\geq 0}$\textit{, }%
namely, 
\begin{equation}
\text{ }\frac{d\tilde{X}_{1}}{\tilde{X}_{1}}=\sigma _{1}\alpha \left( \tilde{%
\lambda}_{1,1}dt+dW_{1}\right) +dY_{1},  \label{relative-metric-1}
\end{equation}%
with 
\begin{equation*}
dY_{1}=-\sigma _{2}\theta _{1}\beta \left( \text{\ }\tilde{\lambda}%
_{1,2}dt+dW_{2}\right) =-\theta _{1}\beta \text{\ }\frac{d\tilde{S}_{1,2}}{%
\tilde{S}_{1,2}},
\end{equation*}%
with $Y_{1,0}=0.$ Note that $Y_{1}$ is driven only by $W_{2}$ and its
dynamics do not depend on $\lambda _{1},\sigma _{1},\alpha .$

Analogous interpretations may be derived for manager $2,$ who now invests in
the ``personalized" fictitious market $\mathcal{\tilde{M}}_{2}^{s}:=\left( B,%
\tilde{S}_{2,1},\tilde{S}_{2,2}\right) $ with (pseudo) stocks $( \tilde{%
S}_{2,1,t})_{t\geq 0}$ and $( \tilde{S}_{2,2,t}) _{t\geq
0},$ solving 
\begin{equation}
\frac{d\tilde{S}_{2,1}}{\tilde{S}_{2,1}}=\sigma _{1}\left( \tilde{\lambda}%
_{2,1}dt+dW_{1}\right) \text{ \  \  \ and \  \ }\frac{d\tilde{S}_{2,2}}{\tilde{S%
}_{2,2}}=\sigma _{2}\left( \text{\ }\tilde{\lambda}_{2,2}dt+dW_{2}\right) ,
\label{pseudo-stocks-2}
\end{equation}%
with modified Sharpe ratios $\tilde{\lambda}_{2,1}$ and \ $\tilde{\lambda}%
_{2,2}$ as in (\ref{lambda-2}). We may then interpret (\ref%
{relative-dynamics2}) as the outcome of investing ratio $\beta $ in stock $%
\tilde{S}_{2,2}$ while maintaining (ratio) allocation $-\theta _{2}\alpha $
in stock $\tilde{S}_{2,1}.$ Alternatively, 
\begin{equation}
\frac{d\tilde{X}_{2}}{\tilde{X}_{2}}=\sigma _{2}\beta \left( \tilde{\lambda}%
_{2,2}dt+dW_{2}\right) +dY_{2},  \label{relative-metric-2}
\end{equation}%
with $\left( Y_{2,t}\right) _{t\geq 0}$ solving%
\begin{equation*}
dY_{2}=-\sigma _{1}\theta _{2}\alpha \left( \tilde{\lambda}%
_{2,1}dt+dW_{1}\right) =-\theta _{2}\alpha \frac{d\tilde{S}_{2,1}}{\tilde{S}%
_{2,1}},
\end{equation*}%
with $Y_{2,0}=0.$

Clearly, the personalized fictitious markets $\mathcal{\tilde{M}}_{1}^{s}$
and $\mathcal{\tilde{M}}_{2}^{s}$ do not coincide due to the asymmetry in
both the competition parameters $\theta _{1}$ and $\theta _{2},$ and the
competitors' allocations $\alpha $ and $\beta $. As the original market $%
\mathcal{M},$ the specialization constraints make both these markets \textit{%
incomplete.} Note also that, in formationally, the markets $\mathcal{M},\mathcal{\tilde{M}}%
_{1}^{s}$ and $\mathcal{\tilde{M}}_{2}^{s}$ do not differ but, conceptually,
forward performance criteria are developed within $\mathcal{\tilde{M}}%
_{1}^{s}$ and $\mathcal{\tilde{M}}_{2}^{s}$.

\subsection{Best-response forward relative performance criterion}

Each manager invests between the riskless asset and the stock in which she
specializes. She also competes with her opponent \textit{passively,} in the
sense that she observes and takes into account the competitor's policy but
without interacting with him. In contrast to all existing settings, however,
the competitor's policy is not a priori modeled, it is only taken to be a
process in the admissible set $\mathcal{A}_{2},$ and is being revealed by
the competitor gradually, as time moves. To model, measure and optimize in
this relative performance setting, we first introduce a suitable criterion.
It extends the original forward criterion, proposed by Musiela and
Zariphopoulou (see \cite{Musiela-first, Musiela-Carmona}) and further
developed by them and others (see, \cite{ElKaroui, MZ-QF, MZ10a, MZ3, Nad-T,
Nad-Z, Zitkovic}).

Throughout, we will be working with the following set of random functions in
the domain $\mathbb{D}=\mathbb{R}_{+}\times \mathbb{R}_{+}.$

\begin{definition}
Let $\mathcal{U}$ be the set of random functions $u\left( z,t\right) ,$ $%
(z,t)\in \mathbb{D}$, such that, for each $t\geq 0$ and $\mathbb{P}$-a.s.,
the mapping $z\rightarrow u\left( z,t\right) $ is strictly concave and
strictly increasing, and $u\left( z,t\right) \in \mathcal{C}^{4,1}$. \textbf{%
\ }
\end{definition}

\begin{definition}
\label{definition of relative forward} Let policy $\beta \in \mathcal{A}%
_{2}. $ An $\mathcal{F}_{t}$-adapted process $\left( V_{1}(\tilde{x}%
_{1},t;\beta )\right) _{t\geq 0},$ $\tilde{x}_{1}\geq 0,$ is called a
best-response forward relative performance criterion for manager 1 if the
following conditions hold:

i) For each $t\geq 0$, $V_{1}(\tilde{x}_{1},t;\beta )\in \mathcal{U}$ $a.s.$

ii)\ For each $\alpha \in \mathcal{A}_{1},$\textit{\ }$V_{1}(\tilde{X}%
_{1},t;\beta )$ is a (local) supermartingale, where $\tilde{X}_{1}$ solves (%
\ref{relative-dynamics1}) with $\alpha $ being used.

iii) There exists $\alpha ^{\ast }\in \mathcal{A}_{1}$, such that $%
V_{1}\left( \tilde{X}_{1}^{\ast },t;\beta \right) $ is a (local) martingale,
where $\tilde{X}_{1}^{\ast }$ solves (\ref{relative-dynamics1}) with $\alpha
^{\ast }$ being used.
\end{definition}

Analogously, we define the best-response forward relative performance for
manager $2$, $\left( V_{2}(\tilde{x}_{2},t;\alpha )\right) _{t\geq 0}$, $%
\tilde{x}_{2}\geq 0$ and $\alpha \in \mathcal{A}_{1}$, requiring that $V_{2}(%
\tilde{X}_{2},t;\alpha )$ and $V_{2}(\tilde{X}_{2}^{\ast },t;\alpha )$ are,
respectively, a (local) supermartingale for any $\beta \in \mathcal{A}_{2}$
and a (local) martingale for an optimal $\beta ^{\ast }\in \mathcal{A}_{2}.$
The notational presence of $\beta $ in $V_{1}$ and $\alpha $ in $V_{2}$ is
self-evident.

In the absence of competition and for It\^{o}-diffusion markets, forward
performance criteria have been constructed also as It\^{o}-diffusion
processes (cf.~\cite{MZ-QF}). However, contrary to the classical expected
utility case, their volatility process is an ``investor-specific" modeling
input. For a chosen volatility process, the supermartingality and
martingality properties impose conditions on the drift of the forward
criterion. Under enough regularity, these conditions lead to the forward
performance SPDE (see \cite{MZ3}), which is a fully nonlinear infinite
dimensional equation. Depending on whether the forward process is
path-dependent or a deterministic functional of stochastic factors, the
forward volatility can be chosen to be path- or state-dependent (see, for
example, \cite{Kallblad, ElKaroui, Liang, Nad-T, Nad-Z, Shkolnikov}). In
general, the underlying problems are inherently ill-posed and extra analysis
is required to identify the viable initial conditions (see, for example, 
\cite{Berrier, MZ10a}).

As mentioned in the introduction, we will work with \textit{locally riskless}
(no volatility) performance processes, 
\begin{equation*}
dV_{1}\left( \tilde{x}_{1},t;\beta \right) =b_{1}\left( \tilde{x}%
_{1},t;\beta \right) dt\text{ \  \  \  \ and \  \ }dV_{2}\left( \tilde{x}%
_{2},t;\alpha \right) =b_{2}\left( \tilde{x}_{2},t;\alpha \right) dt,
\end{equation*}%
for some suitably chosen $\mathcal{F}_{t}$-adapted processes $\left(
b_{1}\left( \tilde{x}_{1},t;\beta \right) \right) _{t\geq 0}$ and $\left(
b_{2}\left( \tilde{x}_{2},t;\alpha \right) \right) _{t\geq 0}.$

Next, we provide a characterization result for the best-response forward
performance criterion $V_{1}\left( \tilde{x}_{1},t;\beta \right) $. Similar
results may be derived for manager $2$ and are, thus, omitted. Throughout,
it is assumed that $\rho ^{2}\neq 1,$ as the case $\rho ^{2}=1$ is more
natural for the asset diversification setting.

\begin{proposition}
\label{forward SPDE 1} Let $\beta \in \mathcal{A}_{2},$ $\rho ^{2}\neq 1,$
and $\tilde{\lambda}_{1,1}$ and $\tilde{\lambda}_{1,2}$ as in (\ref{lamda-1}%
). Consider the random PDE 
\begin{equation}
v_{t}-\frac{1}{2}\tilde{\lambda}_{1,1}^{2}\frac{v_{z}^{2}}{v_{zz}}+\frac{1}{2%
}\left( 1-\rho ^{2}\right) \theta _{1}^{2}\sigma _{2}^{2}\beta
^{2}z^{2}v_{zz}+\left( \rho \tilde{\lambda}_{1,1}-\tilde{\lambda}%
_{1,2}\right) \theta _{1}\sigma _{2}\beta zv_{z}=0,\text{ }  \label{PDE-1}
\end{equation}%
for $\left( z,t\right) \in \mathbb{D},$ and assume that a solution $v\left(
z,t\right) \in \mathcal{U}$ exists, for some admissible initial datum $%
v\left( z,0\right) =V_{1}\left( z,0;\beta \right) .$ Furthermore, let the
process $\left( \alpha _{t}^{\ast }\right) _{t\geq 0}$ be given by 
\begin{equation*}
\alpha ^{\ast }=\alpha ^{\ast }(\tilde{X}_{1}^{\ast },t),
\end{equation*}%
with the random function $\alpha ^{\ast }(z,t),$ $\left( z,t\right) \in 
\mathbb{D},$ defined as 
\begin{equation}
\alpha ^{\ast }(z,t)=\frac{\tilde{\lambda}_{1,1}}{\sigma _{1}}%
R_{1}(z,t)+\rho \frac{\sigma _{2}}{\sigma _{1}}\theta _{1}\beta ,
\label{optimal-passive-general}
\end{equation}%
with 
\begin{equation}
R_{1}\left( z,t\right) :=-\frac{v_{z}(z,t)}{zv_{zz}(z,t)},
\label{risktolerance-passive-1}
\end{equation}%
and $( \tilde{X}_{1,t}^{\ast }) _{t\geq 0}$ solving (\ref{wealths}%
)\ with the control process $\alpha ^{\ast }$ being used. If $\tilde{X}%
_{1}^{\ast }$ is well defined and $\alpha ^{\ast }\in \mathcal{A}_{1}$, then
the process 
\begin{equation*}
V_{1}(\tilde{x}_{1},t;\beta ):=v\left( \tilde{x}_{1},t\right) ,
\end{equation*}%
$\tilde{x}_{1}\geq 0,$ is a locally riskless best-response forward relative
performance criterion and the investment strategy $\alpha ^{\ast }$ is
optimal.
\end{proposition}

\begin{proof}
We first rewrite (\ref{relative-dynamics1})\ as 
\begin{equation}
\frac{d\tilde{X}_{1}}{\tilde{X}_{1}}=\sigma _{1}\hat{\alpha}\left( \tilde{%
\lambda}_{1,1}dt+dW_{1}\right) +\theta _{1}\sigma _{2}\beta \left( \left(
\rho \tilde{\lambda}_{1,1}-\tilde{\lambda}_{1,2}\right) dt-\sqrt{1-\rho ^{2}}%
dW_{1}^{\perp }\right) ,  \label{X-alternative}
\end{equation}%
for $W_{1}^{\perp }$ being a standard Brownian motion orthogonal to $W^{1}$
and the modified policy $\left( \hat{\alpha}_{t}\right) _{t\geq 0},$ 
\begin{equation}
\hat{\alpha}:=\alpha -\rho \theta _{1}\frac{\sigma _{2}}{\sigma _{1}}\beta .
\label{pi-alternative}
\end{equation}%
Assuming that $v\left( z,t\right) \in \mathcal{U}$, Ito's formula yields%
\begin{align*}
dv(\tilde{X}_{1},t)&=v_{t}(\tilde{X}_{1},t)dt+\left( \frac{1}{2}\sigma _{1}^{2}\hat{\alpha}^{2}\tilde{X}_{1}^{2}v_{zz}(%
\tilde{X}_{1},t)+\tilde{\lambda}_{1,1}\hat{\alpha}\tilde{X}_{1}v_{z}(\tilde{X%
}_{1},t)\right) dt\\
&+\left( \frac{1}{2}\left( 1-\rho ^{2}\right) \left( \sigma _{2}\theta
_{1}\beta \right) ^{2}\tilde{X}_{1}^{2}v_{zz}(\tilde{X}_{1},t)+\left( \rho 
\tilde{\lambda}_{1,1}-\tilde{\lambda}_{1,2}\right) \sigma _{2}\theta
_{1}\beta \tilde{X}_{1}v_{z}(\tilde{X}_{1},t)\right) dt\\
&+v_{z}(\tilde{X}_{1},t)\left( \sigma _{1}\tilde{\alpha}dW_{1}-\sigma
_{2}\theta _{1}\beta \sqrt{1-\rho ^{2}}dW_{1}^{\perp }\right).
\end{align*}
Note that for $v_{zz}<0,$ we have 
\begin{equation*}
\frac{1}{2}\sigma _{1}^{2}\hat{\alpha}^{2}\tilde{X}_{1}^{2}v_{zz}(\tilde{X}%
_{1},t)+\tilde{\lambda}_{1,1}\hat{\alpha}\tilde{X}_{1}v_{z}(\tilde{X}%
_{1},t)\leqslant -\frac{1}{2}\tilde{\lambda}_{1,1}^{2}\frac{v_{z}^{2}}{v_{zz}%
},
\end{equation*}%
with the maximum $\hat{\alpha}^{\ast }$ occurring at $\hat{\alpha}^{\ast }=-
\frac{\tilde{\lambda}_{1,1}}{\sigma _{1}}\frac{v_{z}(\tilde{X}_{1},t)}{%
\tilde{X}_{1}v_{zz}(\tilde{X}_{1},t)}.$ The rest of the proof follows easily.
\end{proof}

\bigskip

\textit{Discussion: }Equation (\ref{PDE-1}) is, in general, non-tractable
due to the presence of the second-order linear term $\frac{1}{2}\left(
1-\rho ^{2}\right) \theta _{1}^{2}\sigma _{2}^{2}\beta ^{2}z^{2}v_{zz}$ (the
first-order term $\theta _{1}\left( \rho \tilde{\lambda}_{1,1}-\tilde{\lambda%
}_{1,2}\right) \sigma _{2}\beta zv_{z}$ may be easily absorbed with a mere
time-rescaling). Its form is random and evolves with the market and the
competitor's policy forward in time.

Equations of similar structure also arise in expected utility problems in
the classical setting when there is random endowment and/or labor income processes. To the best of our knowledge, they are also non-tractable and only
general abstract results exist to date (see, among others, \cite{Mos17} and the more recent work \cite{MosSir20}). In the
forward case, an additional complication arises from the ill-posedness of
the problem, for one also needs to specify the class of admissible initial
conditions $V_{1}(\tilde{x}_{1},0;\beta ).$ This is a rather challenging
question, currently investigated by the authors. On the other hand, the CRRA
class provides an example, showing that Definition \ref{definition of relative forward} is not
vacuous.

Despite its non-tractability, equation (\ref{PDE-1}) demonstrates that the
best-response criterion $V_{1}(\tilde{x}_{1},t;\beta )$ is \textit{%
endogenously }specified and depends on the current evolution of the market
and the competitor's policy. Both these features are in contrast to their
analogues in the classical cases.

The optimal policy is constructed through the random feedback functional $%
\alpha ^{\ast }(z,t),$ which consists of the ``myopic"-type term $\frac{%
\tilde{\lambda}_{1,1}}{\sigma _{1}}R_{1}(z,t)$ and the linear term $\rho 
\frac{\sigma _{2}}{\sigma _{1}}\theta _{1}\beta .$ The first component
resembles the one in the original forward setting but now with modified risk
premium $\tilde{\lambda}_{1,1}.$ It depends on the competitor's policy $%
\beta $ through $\tilde{\lambda}_{1,1}$ and $R_{1}(z,t).$ If $\rho \neq 0,$
it may become zero if there exist time(s), say $t_{0},$ such that $\beta
_{t_{0}}=\frac{\lambda _{1,t_{0}}}{\rho \sigma _{2,t_{0}}\theta _{1}}.$

In general, it is difficult to provide any qualitative conclusions on how $%
\alpha ^{\ast }(z,t)$ is influenced by $\beta $ but at least (\ref%
{optimal-passive-general})\ highlights its endogeneity and that it is
affected by the realized market performance, the competitor's policy, and
the manager's realized performance. These characteristics are the outcome of
the flexibility of the normative best-response forward criterion. We stress
that empirical evidence strongly supports such features; see, for example, 
\cite{ChePen09, GorPalPra01, KemRue08}, for the effects of past performance
by the manager and \cite{KemRueTan09} for the impact of realized market
performance. The classical model in which the (terminal) risk tolerance is
exogenously chosen does not seem to capture these phenomena, as argued in
these papers.

Next we note that, in general, $V_{1}(\tilde{x}_{1},t;\beta )$ may \textit{%
not }be time-monotone (albeit being locally riskless). This can be seen from
equation (\ref{PDE-1}) when written as (recall $\rho ^{2}\neq 1)$ 
\begin{equation*}
v_{t}+\frac{1}{2}\left( 1-\rho ^{2}\right) \theta _{1}^{2}v_{zz}\left(
\sigma _{2}\beta z-c_{1}\right) \left( \sigma _{2}\beta z-c_{2}\right) =0,
\end{equation*}%
with $c_{1,2}=\frac{v_{z}}{\theta _{1}v_{zz}}\frac{-\left( \rho \tilde{%
\lambda}_{1,1}-\tilde{\lambda}_{1,2}\right) \pm \sqrt{\Delta }}{1-\rho ^{2}}%
, $ and the process $\left( \Delta _{t}\right) _{t\geq 0}$ given by $\Delta
:=\tilde{\lambda}_{1,1}^{2}-2\rho \tilde{\lambda}_{1,1}\tilde{\lambda}_{1,2}+%
\tilde{\lambda}_{1,2}^{2}>0.$ We easily deduce that $c_{1}c_{2}<0$ and the
lack of time-monotonicity follows from the above equation$\ $and the assumed
spatial concavity of $v.$

We recall that in the absence of competition ($\theta _{1}=0)$, the
analogous locally riskless criterion is given by $u(x_{1},\int_{0}^{t}%
\lambda _{1}^{2}ds),$ with $u$ satisfying $u_{t}=\frac{1}{2}\frac{u_{z}^{2}}{%
u_{zz}},$ $\left( z,t\right) \in \mathbb{D}.$ This process is always
decreasing\textit{\ }in time. The lack of time-monotonicity is one of the
fundamental differences between the forward performance processes with and
without competition, $V_{1}(\tilde{x}_{1},t;\beta )$ and $%
u(x_{1},\int_{0}^{t}\lambda _{1}^{2}ds)$. We comment more on this in the
next section.

If $\rho =0$, then $\tilde{\lambda}_{1,1}=\lambda _{1}$ and equation (\ref%
{PDE-1}) reduces to%
\begin{equation*}
v_{t}-\frac{1}{2}\lambda _{1}^{2}\frac{v_{z}^{2}}{v_{zz}}+\frac{1}{2}\theta
_{1}^{2}\sigma _{2}^{2}\beta ^{2}z^{2}v_{zz}-\tilde{\lambda}_{1,2}\theta
_{1}\sigma _{2}\beta zu_{z}=0.
\end{equation*}%
In turn, $\alpha ^{\ast }(z,t)=-\frac{\lambda _{1}}{\sigma _{1}}\frac{%
v_{z}(z,t)}{zv_{zz}(z,t)},$ with $v$ still depending on $\beta $ through the
coefficients in the reduced equation above.

If $\rho \neq 0,$ relative performance concerns might lead to zero
allocation in $\tilde{S}_{1,1},$ at time(s) $t_{0}$ such that $\frac{\tilde{%
\lambda}_{1,1,t_{0}}}{\sigma _{1,t_{0}}}R_{1}(z,t_{0})+\rho \frac{\sigma
_{2,t_{0}}}{\sigma _{1,t_{0}}}\theta _{1}\beta _{t_{0}}=0.$

\subsubsection{The CRRA case}

\label{subsec:CRRA passive}

To provide further insights on the forward relative performance criteria and
also compare them with the ones in the classical setting, we study the case
of homothetic criteria for manager $1.$ We impose no assumption on what
criterion manager $2$ might follow, we only use assume that she follows an
arbitrary policy $\beta \in \mathcal{A}_{2}$.

\begin{proposition}
\label{passive-specialization} Let policy $\beta \in \mathcal{A}_{2},$ $\rho
^{2}\neq 1,$ and $\tilde{\lambda}_{1,1}$ and $\tilde{\lambda}_{1,2}$ as in (%
\ref{lamda-1}). Let $\gamma _{1}>0,$ $\gamma _{1}\neq 0,$ and $\left( \eta
_{1,t}\right) _{t\geq 0}$ be given by 
\begin{equation}
\eta _{1}=\tilde{\lambda}_{1,1}^{2}+2\left( \rho \tilde{\lambda}_{1,1}-%
\tilde{\lambda}_{1,2}\right) \theta _{1}\sigma _{2}\beta \gamma _{1}-\left(
1-\rho ^{2}\right) \theta _{1}^{2}\sigma _{2}^{2}\beta ^{2}\gamma _{1}^{2}.
\label{eta-1}
\end{equation}%
Then, the process 
\begin{equation}
V_{1}\left( \tilde{x}_{1},t;\beta \right) =\frac{\tilde{x}_{1}^{1-\gamma
_{1}}}{1-\gamma _{1}}e^{-\int_{0}^{t}\frac{1-\gamma _{1}}{2\gamma _{1}}\eta
_{1}ds},\text{ \  \  \ }  \label{CRRA-passive}
\end{equation}%
is a locally riskless best-response forward criterion and the investment
policy%
\begin{equation}
\alpha ^{\ast }=\frac{1}{\gamma _{1}}\frac{\tilde{\lambda}_{1,1}}{\sigma _{1}%
}+\rho \theta _{1}\frac{\sigma _{2}}{\sigma _{1}}\beta
\label{portoflio-passive}
\end{equation}%
is optimal.\textbf{\ }
\end{proposition}

\begin{proof}
We look for candidate criteria of the separable form $V_{1}(\tilde{x}%
_{1},t;\beta )=\frac{\tilde{x}^{1-\gamma _{1}}}{1-\gamma _{1}}K,$ where $%
\left( K_{t}\right) _{t\geq 0}$ is an $\mathcal{F}_{t}$-adapted process,
differentiable in $t$ with $K_{0}=1$. Using equation (\ref{PDE-1}), the
boundedness assumption on the Sharpe ratios and the admissibility of $\beta
, $ we easily conclude.
\end{proof}

We may rewrite the process $\left( \eta _{1,t}\right) _{t\geq 0}$ as 
\begin{equation}
\eta _{1}=\left( \lambda _{1,1}-\delta _{1}\theta _{1}\sigma _{2}\beta
\right) ^{2}+\left( \rho ^{2}(1-\gamma _{1})^{2}+\gamma _{1}(1-\gamma _{1}+%
\frac{1}{\theta _{1}})-\delta _{1}^{2}\right) \theta _{1}^{2}\sigma
_{2}^{2}\beta ^{2},  \label{Basak-1}
\end{equation}%
with $\left( \delta _{1,t}\right) _{t\geq 0}$ given by 
\begin{equation}
\delta _{1}=\gamma _{1}\frac{\lambda _{2}}{\lambda _{1}}+\rho (1-\gamma
_{1}).  \label{Basak-2}
\end{equation}%
Similar expressions were derived in \cite{Basak15} for the special case of
log-normal markets for power utilities in the classical setting. Herein, we
have analogous results for general $\mathcal{F}_{t}$-adapted processes $\eta
_{1}$ and $\delta _{1}.$ We stress that \textit{no} solutions of form (\ref%
{CRRA-passive}), (\ref{Basak-1}) and (\ref{Basak-2})\ may be derived in the
classical setting beyond the log-normal case.

The criterion $V_{1}(\tilde{x}_{1},t;\beta )$ resembles its forward
counterpart in the absence of relative performance ($\theta _{1}=0)$, given
by $u\left( x_{1},t\right) =\frac{x_{1}^{1-\gamma _{1}}}{1-\gamma _{1}}%
e^{-\int_{0}^{t}\frac{1-\gamma _{1}}{2\gamma _{1}}\lambda _{1}^{2}ds}$ (see \cite{MZ10a}), which is however always time-monotone.

Rewriting (\ref{portoflio-passive}) as 
\begin{equation}
\alpha ^{\ast }=\frac{1}{\gamma _{1}}\frac{\lambda _{1}}{\sigma _{1}}+\rho
\theta _{1}\left( 1-\frac{1}{\gamma _{1}}\right) \frac{\sigma _{2}}{\sigma
_{1}}\beta ,  \label{a-passive-alternative}
\end{equation}%
we see that depending on the sign of the various terms, manager $1$ might
invest more or less in the risky asset under relative performance concerns.
For example, for $\rho >0$, $\frac{\sigma _{2}}{\sigma _{1}}>0$, and a long
competitor's strategy, $\beta >0,$ we have $\rho \theta _{1}\left( 1-\frac{1%
}{\gamma _{1}}\right) \frac{\sigma _{2}}{\sigma _{1}}\beta >0$ if $\gamma
_{1}<0$, while $\rho \theta _{1}\left( 1-\frac{1}{\gamma _{1}}\right) \frac{%
\sigma _{2}}{\sigma _{1}}\beta <0$ if $0<\gamma _{1}<1.$ These results are
also consistent with the ones in \cite{Basak15} but, now, for a much more
flexible framework. Finally, if the market price of risk $\lambda _{1}$
increases, the position on the familiar asset always increases even with
relative performance concerns. This is consistent with the fact that when
the performance of the asset the manager invests in improves, she tends to
increase her position to it. The process $\lambda _{1}$ usually refers to
the manager's active-management \textit{ability} (see among others \cite%
{Sha66}).

Symmetric results are deduced for manager $2$ if her competitor follows
policy $\alpha \in \mathcal{A}_{1}.$ Namely, for $\gamma _{2}>0,\gamma
_{2}\neq 1,$ and $\left( \eta _{2,t}\right) _{t\geq 0}$ with

\begin{equation}
\eta _{2}:=\tilde{\lambda}_{2,2}^{2}+2\left( -\tilde{\lambda}_{2,1}+\rho 
\tilde{\lambda}_{1,1}\right) \sigma _{1}\theta _{2}\alpha \gamma _{2}-\left(
1-\rho ^{2}\right) \left( \sigma _{1}\theta _{2}\alpha \right) ^{2}\gamma
_{2}^{2},  \label{eta-2}
\end{equation}%
the process $\left( V_{2}\left( \tilde{x}_{2},t;\alpha \right) \right)
_{t\geq 0}$ given by 
\begin{equation}
V_{2}\left( \tilde{x}_{2},t;\alpha \right) =\frac{\tilde{x}_{2}^{1-\gamma
_{2}}}{1-\gamma _{2}}e^{-\int_{0}^{t}\frac{1-\gamma _{2}}{2\gamma _{2}}\eta
_{2}ds},\text{ \  \  \ }  \label{V-2}
\end{equation}%
is a locally riskless best-response forward criterion and the investment
policy%
\begin{equation}
\beta ^{\ast }=\frac{1}{\gamma _{2}}\frac{\tilde{\lambda}_{2,2}}{\sigma _{2}}%
+\rho \theta _{2}\frac{\sigma _{1}}{\sigma _{2}}\alpha =\frac{1}{\gamma _{2}}%
\frac{\lambda _{2}}{\sigma _{2}}+\rho \theta _{2}\left( 1-\frac{1}{\gamma
_{2}}\right) \frac{\sigma _{2}}{\sigma _{1}}\alpha  \label{beta-2}
\end{equation}%
is optimal.

Finally, we may construct a best-response (locally riskless)\ forward
criterion for the limiting cases $\gamma _{1}=0$ and/or $\gamma _{2}=0.$
Looking for a candidate process of the additive form $V_{1}(\tilde{x}%
_{1},t;\beta )=\log \tilde{x}_{1}+K,$ for a suitable process $\left(
K_{t}\right) _{t\geq 0},$ equation (\ref{PDE-1}) yields 
\begin{equation*}
V_{1}(\tilde{x}_{1},t;\beta )=\log \tilde{x}_{1}+\int_{0}^{t}\left( \frac{1}{%
2}\tilde{\lambda}_{1,1}^{2}-\left( \rho \tilde{\lambda}_{1,1}-\tilde{\lambda}%
_{1,2}\right) \theta _{1}\sigma _{2}\beta +\frac{1}{2}\left( 1-\rho
^{2}\right) \theta _{1}^{2}\left( \sigma _{2}\beta \right) ^{2}\right) ds,
\end{equation*}%
with optimal policy $ \alpha^{\ast }=\frac{%
\tilde{\lambda}_{1,1}}{\sigma _{1}}+\rho \frac{\sigma _{2}}{\sigma _{1}}%
\theta _{1}\beta .$ Similar results can be produced for the case $\gamma
_{2}=0.$

\subsection{Forward Nash equilibrium}

The asset managers not only trade between the riskless account and the
respective specialized risky asset but, also, interact dynamically with each
other. Then, the individual best-response problems lead conceptually to a
pure-strategy Nash game. We call the equilibrium of this game a \textit{%
forward Nash equilibrium} and propose the following definition for its
analysis.

We recall the modified risk premia $\tilde{\lambda}_{1,1}\left( \beta
\right) ,\tilde{\lambda}_{1,2}\left( \beta \right) $ and $\tilde{\lambda}%
_{2,1}\left( \alpha \right) ,\tilde{\lambda}_{2,2}\left( \alpha \right) $
(cf.~(\ref{lamda-1}) and (\ref{lambda-2})), highlighting their dependence on
the competitor's policies.

\begin{definition}
\label{definition of forward Nash} A forward Nash equilibrium consists of
two pairs of $\mathcal{F}_{t}$-adapted processes,\\
 $\left( V_{1}\left( \tilde{%
x}_{1},t;\beta ^{\ast }\right) _{t\geq 0},\left( \alpha _{t}^{\ast }\right)
_{t\geq 0}\right) $ and $(\left( V_{2}\left( \tilde{x}_{2},t;\alpha ^{\ast
}\right) \right) _{t\geq 0},\left( \beta _{t}^{\ast }\right) _{t\geq 0}),$ $%
\tilde{x}_{1},\tilde{x}_{2}>0,$ $t\geq 0,$ with the following properties:

i) The processes $\alpha ^{\ast }\in \mathcal{A}_{1}$ and $\beta ^{\ast }\in 
\mathcal{A}_{2},$

ii) The processes $V_{1}\left( \tilde{x}_{1},t;\beta ^{\ast }\right) $, $%
V_{2}\left( \tilde{x}_{2},t;\alpha ^{\ast }\right) \in \mathcal{U}.$

ii) For $\alpha \in \mathcal{A}_{1},$ $V_{1}(\tilde{X}_{1},t;\beta ^{\ast })$
is a (local) super-martingale and $V_{1}(\tilde{X}_{1,t}^{\ast },t;\beta
^{\ast })$ is a (local) martingale where $\tilde{X}_{1}$ and $\tilde{X}%
_{1}^{\ast }$ solve (\ref{relative-dynamics1}) with $\tilde{\lambda}_{1,1}=%
\tilde{\lambda}_{1,1}\left( \beta ^{\ast }\right) $ and $\tilde{\lambda}%
_{1,2}=\tilde{\lambda}_{1,2}\left( \beta ^{\ast }\right) ,$ and with $\alpha 
$ and $\alpha ^{\ast }$ being, respectively, used.

iii) For $\beta \in \mathcal{A}_{2},$ $V_{2}(\tilde{X}_{2},t;\alpha ^{\ast
}) $ is a (local) super-martingale and $V_{2}(\tilde{X}_{2,t}^{\ast
},t;\alpha ^{\ast })$ is a (local) martingale where $\tilde{X}_{2}$ and $%
\tilde{X}_{2}^{\ast }$ solve (\ref{relative-dynamics2}) with $\tilde{\lambda}%
_{2,1}=\tilde{\lambda}_{2,1}\left( \alpha ^{\ast }\right) $ and $\tilde{%
\lambda}_{2,2}=\tilde{\lambda}_{2,2}\left( \alpha ^{\ast }\right) ,$ and
with $\beta $ and $\beta ^{\ast }$ being, respectively, used.
\end{definition}

If, under appropriate integrability conditions, the processes $\left( V_{1}(%
\tilde{X}_{1},t;\beta ^{\ast })\right) _{t\geq 0}$ and $\left( V_{1}(\tilde{X%
}_{1}^{\ast },t;\beta ^{\ast })\right) _{t\geq 0}$ are, respectively, a true
supermartingale and a true martingale then, for any $\alpha \in \mathcal{A}%
_{1}$, 
\begin{equation*}
\mathbb{E}\left[ V_{1}(\tilde{X}_{1}^{\ast },t;\beta ^{\ast })\right] =%
\mathbb{E}\left[ V_{1}(\tilde{x}_{1},0)\right] \geq \mathbb{E}\left[ V_{1}(%
\tilde{X}_{1}^{1},t;\beta ^{\ast })\right] .
\end{equation*}%
Analogously,

\begin{equation*}
\mathbb{E}\left[ V_{2}(\tilde{X}_{2}^{\ast },t;\alpha ^{\ast })\right] =%
\mathbb{E}\left[ V_{2}(\tilde{x}_{2},0)\right] \geq \mathbb{E}\left[ V_{2}(%
\tilde{X}_{2},t;\alpha ^{\ast })\right] .
\end{equation*}%
In other words, no unilateral deviation in strategy by either manager will
result in an increase in the expected utility of her relative performance
metric.

From Proposition \ref{forward SPDE 1} and, in particular, the best-response
strategy (\ref{optimal-passive-general}) and analogous results for the
optimal policy $\beta ^{\ast }$, it follows that the candidate forward Nash
equilibrium strategies should satisfy the system of equations 
\begin{equation}
\begin{cases}
\alpha ^{\ast }=\frac{\tilde{\lambda}_{1,1}\left( \beta ^{\ast }\right) }{%
\sigma _{1}}R_{1}^{\ast }\left( \tilde{X}_{1}^{\ast },t;\beta ^{\ast
}\right) +\rho \theta _{1}\frac{\sigma _{2}}{\sigma _{1}}\beta ^{\ast } \\ 
\\ 
\beta ^{\ast }=\frac{\tilde{\lambda}_{2,2}\left( \alpha ^{\ast }\right) }{%
\sigma _{2}}R_{2}^{\ast }\left( \tilde{X}_{2}^{\ast },t;\alpha ^{\ast
}\right) +\rho \theta _{2}\frac{\sigma _{1}}{\sigma _{2}}\alpha ^{\ast },%
\end{cases}
\label{system}
\end{equation}%
where $\left( R_{1,t}^{\ast }\left( \tilde{X}_{1}^{\ast },t;\beta ^{\ast
}\right) \right) _{t\geq 0}$ \ and $\left( R_{2,t}^{\ast }\left( \tilde{X}%
_{2}^{\ast },t;\alpha ^{\ast }\right) \right) _{t\geq 0}$ are defined as 
\begin{equation*}
R_{1}^{\ast }\left( \tilde{X}_{1}^{\ast },t;\beta ^{\ast }\right) =-\frac{%
v_{1,z}(\tilde{X}_{1}^{\ast },t)}{\tilde{X}_{1}^{\ast }v_{1,zz}(\tilde{X}%
_{1}^{\ast },t)}\text{ \  \  \ and \  \  \ }R_{2}^{\ast }\left( \tilde{X}%
_{2}^{\ast },t;\alpha ^{\ast }\right) =-\frac{v_{2,z}(\tilde{X}_{2}^{\ast
},t)}{\tilde{X}_{2}^{\ast }v_{2,zz}(\tilde{X}_{2}^{\ast },t)},
\end{equation*}%
with $v_{1}\left( z,t\right) $ and $v_{2}\left( z,t\right) ,$ $\left(
z,t\right) \in \mathbb{D},$ solving%
\begin{equation}
v_{1,t}-\frac{1}{2}\tilde{\lambda}_{1,1}^{2}\left( \beta ^{\ast }\right) 
\frac{v_{1,z}^{2}}{v_{1,zz}}+\frac{1}{2}\left( 1-\rho ^{2}\right) \sigma
_{2}^{2}\theta _{1}^{2}\beta ^{\ast 2}z^{2}v_{1,zz}+\left( \rho \tilde{%
\lambda}_{1,1}\left( \beta ^{\ast }\right) -\tilde{\lambda}_{1,2}\left(
\beta ^{\ast }\right) \right) \sigma _{2}\theta _{1}\beta ^{\ast }zv_{1,z}=0
\label{u-1}
\end{equation}%
and 
\begin{equation}
v_{2,t}-\frac{1}{2}\tilde{\lambda}_{2,2}^{2}\left( \alpha ^{\ast }\right) 
\frac{v_{2,z}^{2}}{v_{2,zz}}+\frac{1}{2}\left( 1-\rho ^{2}\right) \sigma
_{1}^{2}\theta _{2}^{2}\alpha ^{\ast 2}z^{2}v_{2,zz}+\left( -\tilde{\lambda}%
_{2,1}\left( \alpha ^{\ast }\right) +\rho \tilde{\lambda}_{2,2}\left( \alpha
^{\ast }\right) \right) \sigma _{1}\theta _{2}\alpha ^{\ast }zv_{2,z}=0.
\label{u-2}
\end{equation}%
System (\ref{system}) is in general non-tractable because of the highly
non-linear terms $R_{1}^{\ast }\left( \tilde{X}_{1}^{\ast },t;\beta ^{\ast
}\right) $ and $R_{2}^{\ast }\left( \tilde{X}_{2}^{\ast },t;\alpha ^{\ast
}\right) .$

\subsubsection{The CRRA cases}

For tractability and to highlight the differences between the forward
approach and the classical setting, we examine the case of homothetic
criteria for \textit{both }managers.

\begin{proposition}
Let $\gamma _{1},\gamma _{2}>0$ with $\gamma _{1},\gamma _{2}\neq 1,$ and
assume that 
\begin{equation}
\delta :=\gamma _{1}\gamma _{2}-\rho ^{2}\theta _{1}\theta _{2}(1-\gamma
_{1})(1-\gamma _{2})\neq 0.  \label{delta-constant}
\end{equation}%
Consider the processes $\left( \alpha _{t}^{\ast }\right) _{t\geq 0},$ $%
\left( \beta _{t}^{\ast }\right) _{t\geq 0}$ given by 
\begin{equation}
\alpha ^{\ast }=\frac{\gamma _{2}\lambda _{1}-\rho \theta _{1}(1-\gamma
_{1})\lambda _{2}}{\sigma _{1}\delta }\text{ \  \  \ and \  \  \ }\beta ^{\ast }=%
\frac{\gamma _{1}\lambda _{2}-\rho \theta _{2}(1-\gamma _{2})\lambda _{1}}{%
\sigma _{2}\delta }.  \label{Nash-optimal}
\end{equation}%
Let also $\left( \eta _{1,t}^{\ast }\right) _{t\geq 0}$ and $\left( \eta
_{2,t}^{\ast }\right) _{t\geq 0}$ be given by (\ref{eta-1}) and (\ref{eta-2}%
) when $\beta ^{\ast }$ and $\alpha ^{\ast }$are, respectively, used and $%
\left( V_{1}\left( \tilde{x}_{1},t;\beta ^{\ast }\right) \right) _{t\geq 0}$
and $\left( V_{2}\left( \tilde{x}_{2},t;\alpha ^{\ast }\right) \right)
_{t\geq 0}$ defined as%
\begin{equation*}
V_{1}\left( \tilde{x}_{1},t;\beta ^{\ast }\right) =\frac{\tilde{x}%
_{1}^{1-\gamma _{1}}}{1-\gamma _{1}}e^{-\int_{0}^{t}\frac{1-\gamma _{1}}{%
2\gamma _{1}}\eta _{1}^{\ast }ds},\text{ }V_{2}\left( \tilde{x}_{2},t;\alpha
^{\ast }\right) =\frac{\tilde{x}_{2}^{1-\gamma _{2}}}{1-\gamma _{2}}%
e^{-\int_{0}^{t}\frac{1-\gamma _{2}}{2\gamma _{2}}\eta _{2}^{\ast }ds}.
\end{equation*}%
Then, the pair of processes $(\left( V_{1}\left( \tilde{x}_{1},t;\beta
^{\ast }\right) ,\alpha ^{\ast }\right) $ and $\left( V_{2}\left( \tilde{x}%
_{2},t;\alpha ^{\ast }\right) ,\beta ^{\ast }\right) $ constitutes a forward
Nash equilibrium.
\end{proposition}

\begin{proof}
From (\ref{a-passive-alternative}) and (\ref{beta-2}), we deduce that the
candidate strategies $\left( \alpha _{t}^{\ast }\right) _{t\geq 0},\left(
\beta _{t}^{\ast }\right) _{t\geq 0}${\ }must solve the system 
\begin{equation*}
\begin{cases}
\alpha {^{\ast }-\left( 1-\frac{1}{\gamma _{1}}\right) \rho \theta _{1}\frac{%
\sigma _{2}}{\sigma _{1}}\beta ^{\ast }=\frac{1}{\gamma _{1}}\frac{\lambda
_{1}}{\sigma _{1}}} \\ 
\\ 
{-\left( 1-\frac{1}{\gamma _{2}}\right) \rho \theta _{2}\frac{\sigma _{1}}{%
\sigma _{2}}\alpha ^{\ast }+\beta ^{\ast }=\frac{1}{\gamma _{2}}\frac{%
\lambda _{2}}{\sigma _{2}}.}%
\end{cases}%
\end{equation*}%
Using that its determinant is given by $\frac{\delta }{\gamma _{1}\gamma _{2}%
},$ with $\delta $ as in (\ref{delta-constant}) and, by assumption, $\delta
\neq 0$, we easily deduce (\ref{Nash-optimal}). Furthermore, $\alpha ^{\ast
}\in \mathcal{A}_{1}$ and $\beta ^{\ast }\in \mathcal{A}_{2},$ given the
assumption on bounded $\lambda _{1}$ and $\lambda _{2}.$ The rest of the
proof follows easily.
\end{proof}

In the special case $\rho =0,$ the forward Nash equilibrium strategies
simplify to, 
\begin{equation*}
\alpha ^{\ast }=\frac{1}{\gamma _{1}}\frac{\lambda _{1}}{\sigma _{1}}\text{
\  \  \ and \  \ }\beta _{2}^{\ast }=\frac{1}{\gamma _{2}}\frac{\lambda _{2}}{%
\sigma _{2}},
\end{equation*}%
which are the optimal policies without competition. Note, however, that the
associated forward Nash criteria still depend on the other manager's
strategy through the processes $\eta _{1}^{\ast }$ and $\eta _{2}^{\ast }$.

Continuing the discussion in 2.1.1., we mention that the forward Nash
equilibrium investment strategies have the same form as those of the
classical setting in a log-normal market (see \cite[Proposition 1]{Basak15}). Hence, we may generalize \textit{all} comparative statics of \cite{Basak15} in the general It\^{o}-diffusion market setting herein.

\section{Asset diversification and forward competition}

\label{sec:diversification}

In this section we impose the situation where both managers invest in the same market $\mathcal{M}=\left(
B,S_{1},S_{2}\right) ,$ with $S_{1},S_{2}$ solving (\ref{stocks}) and
without any trading constraints. This case is particularly popular when
managers aim to beat the same benchmark. The managers
have relative performance concerns and may interact passively or
competitively. As in the asset specialization case, we incorporate these
concerns by working with relative wealth processes with competition
parameters $\theta _{1},\theta _{2}.$ We measure the performance of their strategies using forward best response and forward Nash equilibrium
criteria, respectively. We define them as in Definitions \ref{definition of relative forward} and \ref{definition of forward Nash}, and we also work with locally riskless processes.

Using (\ref{stocks}), the (discounted by the bond) wealth processes $\left(
X_{1,t}\right) _{t\geq 0}$ and $\left( X_{2,t}\right) _{t\geq 0}$, $t\geq 0,$
satisfy 
\begin{equation}
\frac{dX_{1}}{X_{1}}=\sigma _{1}\alpha _{1}(\lambda _{1}dt+dW_{1})+\sigma
_{2}\alpha _{2}(\lambda _{2}dt+dW_{2})\text{ }  \label{wealth-1-div}
\end{equation}%
and 
\begin{equation}
\frac{dX_{2}}{X_{2}}=\sigma _{1}\beta _{1}(\lambda _{1}dt+dW_{1})+\sigma
_{2}\beta _{2}(\lambda _{2}dt+dW_{2})\text{,}  \label{wealth-2-div}
\end{equation}%
with $X_{1,0}=x_{1}>0$ and $X_{2,0}=x_{2}>0,$ and $\alpha _{1},\alpha _{2}$
(resp.~$\beta _{1},\beta _{2}$) being the fractions of wealth $X_{1}$ (resp.~$X_{2}$) invested in asset classes $S_{1}$ and $S_{2},$ respectively. The
set $\mathcal{A}$ of admissible policies $\alpha =\left( \alpha _{1},\alpha
_{2}\right) $ and $\beta =$ $\left( \beta _{1},\beta _{2}\right) $ is
defined similarly to (\ref{A}).

For $\theta _{1},\theta _{2}\in \left( 0,1\right] ,$ $\alpha =\left( \alpha
_{1},\alpha _{2}\right) $ and $\beta =\left( \beta _{1},\beta _{2}\right) ,$
the biased benchmark processes $\left( X_{1,t}^{\theta _{2}}\right) _{t\geq
0}$ and $\left( X_{2,t}^{\theta _{1}}\right) _{t\geq 0}$ solve 
\begin{equation*}
\frac{dX_{1,t}^{\theta _{2}}}{X_{1,t}^{\theta _{2}}}=\theta _{2}\sigma
_{1}\alpha _{1}(\lambda _{1}dt+dW_{1})+\text{\ }\theta _{2}\sigma _{2}\alpha
_{2}(\lambda _{2}dt+dW_{2})+\frac{1}{2}\theta _{2}(\theta _{2}-1)C_{1}\left(
\alpha \right) dt
\end{equation*}%
with the process $\left( C_{1,t}\left( \alpha \right) \right) _{t\geq 0}$,%
\begin{equation}
C_{1}\left( \alpha \right) :=\sigma _{1}^{2}\alpha _{1}^{2}+2\rho \sigma
_{1}\sigma _{2}\alpha _{1}\alpha _{2}+\sigma _{2}^{2}\alpha _{2}^{2}.
\label{C-1}
\end{equation}%
Similarly, 
\begin{equation*}
\frac{dX_{2,t}^{\theta _{1}}}{X_{2,t}^{\theta _{1}}}=\theta _{1}\sigma
_{1}\beta _{1}(\lambda _{1}dt+dW_{1})+\text{\ }\theta _{1}\sigma _{2}\beta
_{2}(\lambda _{2}dt+dW_{2})+\frac{1}{2}\theta _{1}(\theta _{1}-1)C_{2}\left(
\beta \right) dt,
\end{equation*}%
with the process $\left( C_{2,t}\left( \beta \right) \right) _{t\geq 0},$ 
\begin{equation}
C_{2}\left( \beta \right) =\sigma _{1}^{2}\beta _{1}^{2}+2\rho \sigma
_{1}\sigma _{2}\beta _{1}\beta _{2}+\sigma _{2}^{2}\beta _{2}^{2}.
\label{C-2}
\end{equation}

Direct calculations yield that the relative wealths processes $\tilde{X}%
_{1}:=\frac{X_{1}}{X_{2,t}^{\theta _{1}}}$ and $\tilde{X}_{2}:=$ $\frac{X_{2}%
}{X_{1,t}^{\theta _{2}}}$ satisfy 
\begin{align}
\nonumber \frac{d\tilde{X}_{1}}{\tilde{X}_{1}}&=\sigma _{1}\alpha _{1}\left( \tilde{%
\lambda}_{1,1}dt+dW_{1}\right) +\sigma _{2}\alpha _{2}\left( \tilde{\lambda}
_{1,2}dt+dW_{2}\right)\\
&-\sigma _{1}\theta _{1}\beta _{1}\left( \lambda _{1}dt+dW_{1}\right) -\sigma
_{2}\theta _{1}\beta _{2}\left( \lambda _{2}dt+dW_{2}\right) +\frac{1}{2}%
\theta _{1}\left( 1+\theta _{1}\right) C_{2}\left( \beta \right) dt,
\label{X-1-div}
\end{align}
with the processes $\left( \tilde{\lambda}_{1,1,t}\right) _{t\geq 0}$ and $%
\left( \tilde{\lambda}_{1,2,t}\right) _{t\geq 0},$ 
\begin{equation}
\tilde{\lambda}_{1,1}:=\lambda _{1}-\theta _{1}\left( \sigma _{1}\beta
_{1}+\rho \sigma _{2}\beta _{2}\right) \text{ \  \  \ and \  \  \ }\tilde{\lambda%
}_{1,2}:=\lambda _{2}-\theta _{1}\left( \rho \sigma _{1}\beta _{1}+\sigma
_{2}\beta _{2}\right) .  \label{lamda-1-^}
\end{equation}%
Similarly,
\begin{align}
\nonumber\frac{d\tilde{X}_{2}}{\tilde{X}_{2}}&=\sigma _{1}\beta _{1}\left( \tilde{%
\lambda}_{2,1}dt+dW_{1}\right) +\sigma _{2}\beta _{2}\left( \tilde{\lambda}%
_{2,2}dt+dW_{2}\right)\\
&-\sigma _{1}\theta _{2}\alpha _{1}\left( \lambda _{1}dt+dW_{1}\right)
-\sigma _{2}\theta _{2}\alpha _{2}\left( \lambda _{2}dt+dW_{2}\right) +\frac{%
1}{2}\theta _{2}\left( 1+\theta _{2}\right) C_{1}\left( \alpha \right) dt,
\label{X-2-div}
\end{align}
with the processes $\left( \tilde{\lambda}_{2,1,t}\right) _{t\geq 0}$ and $%
\left( \tilde{\lambda}_{2,2,t}\right) _{t\geq 0}$, 
\begin{equation}
\tilde{\lambda}_{2,1}:=\lambda _{1}-\theta _{2}\left( \sigma _{1}\alpha
_{1}+\rho \sigma _{2}\alpha _{2}\right) \text{ \  \  \ and \  \ }\tilde{\lambda}%
_{2,2}:=\lambda _{2}-\theta _{2}\left( \rho \sigma _{1}\alpha _{1}+\sigma
_{2}\alpha _{2}\right) .  \label{lambda-2-^}
\end{equation}%
As in the asset specialization case, we may interpret (\ref{X-1-div}) as the
wealth of a manager who invests in the personalized fictitious market $%
\mathcal{\tilde{M}}_{1}^{d}:=\left( B,\tilde{S}_{1,1},\tilde{S}_{1,2}\right) 
$ with (pseudo) stocks $\tilde{S}_{1,1},\tilde{S}_{1,2}$ solving 
\begin{equation*}
\frac{d\tilde{S}_{1,1}}{\tilde{S}_{1,1}}=\sigma _{1}\left( \tilde{\lambda}%
_{1,1}dt+dW_{1}\right) \  \  \  \text{and }\  \  \frac{d\tilde{S}_{1,2}}{\tilde{S}%
_{1,2}}=\sigma _{2}\left( \text{\ }\tilde{\lambda}_{1,2}dt+dW_{2}\right) ,
\end{equation*}%
with $\tilde{\lambda}_{1,1}$ and $\tilde{\lambda}_{1,2}$ given in (\ref%
{lamda-1-^}), while receiving returns from a \textit{random endowment process} $%
\left( Y_{1,t}\right) _{t\geq 0},$%
\begin{equation*}
\frac{d\tilde{X}_{1}}{\tilde{X}_{1}}=\alpha _{1}\sigma _{1}\left( \tilde{%
\lambda}_{1,1}dt+dW_{1}\right) +\alpha _{2}\sigma _{2}\left( \tilde{\lambda}%
_{1,2}dt+dW_{2}\right) +dY_{1,t}
\end{equation*}%
with%
\begin{equation*}
dY_{1}=-\theta _{1}\sigma _{1}\beta _{1}\left( \lambda _{1}dt+dW_{1}\right)
-\theta _{1}\sigma _{2}\beta _{2}\left( \lambda _{2}dt+dW_{2}\right) +\frac{1%
}{2}\theta _{1}\left( 1+\theta _{1}\right) C_{2}\left( \beta \right) dt
\end{equation*}%
and $Y_{1,0}=0.$

Similarly, manager $2$ invests in a personalized fictitious market $\mathcal{%
\tilde{M}}_{2}^{d}:=\left( B,\tilde{S}_{2,1},\tilde{S}_{2,2}\right) $ with
(pseudo) stocks $\tilde{S}_{2,1},\tilde{S}_{2,2}$ solving 
\begin{equation*}
\frac{d\tilde{S}_{2,1}}{\tilde{S}_{2,1}}=\sigma _{2}\left( \tilde{\lambda}%
_{2,1}dt+dW_{1}\right) \  \  \  \text{and }\  \  \frac{d\tilde{S}_{2,2}}{\tilde{S}%
_{2,2}}=\sigma _{2}\left( \text{\ }\tilde{\lambda}_{2,2}dt+dW_{2}\right) ,
\end{equation*}%
with $\tilde{\lambda}_{2,1}$ and $\tilde{\lambda}_{2,2}$ given in (\ref%
{lambda-2-^}), and 
\begin{equation*}
\frac{d\tilde{X}_{2}}{\tilde{X}_{2}}=\sigma _{1}\beta _{1}\left( \tilde{%
\lambda}_{2,1}dt+dW_{1}\right) +\sigma _{2}\beta _{2}\left( \tilde{\lambda}%
_{2,2}dt+dW_{2}\right) +dY_{2}
\end{equation*}%
with 
\begin{equation*}
dY_{2}=-\theta _{2}\sigma _{1}\alpha _{1}\left( \lambda _{1}dt+dW_{1}\right)
-\theta _{2}\sigma _{2}\alpha _{2}\left( \lambda _{2}dt+dW_{2}\right) +\frac{%
1}{2}\theta _{2}\left( 1+\theta _{2}\right) C_{1}\left( \alpha \right) dt,
\end{equation*}%
and $Y_{2,0}=0.$

The personalized fictitious markets $\mathcal{\tilde{M}}_{1}^{d}$ and $%
\mathcal{\tilde{M}}_{2}^{d}$ are both \textit{complete, }in contrast to
their counterparts $\mathcal{\tilde{M}}_{1}^{s}$ and $\mathcal{\tilde{M}}%
_{2}^{s}$ in the asset specialization case. This completeness makes the
underlying problems tractable, as we discuss next.

\subsection{Best-response forward relative performance criteria}

\label{Passive forward performance criteria}

In analogy to the asset diversification case, we apply Definition \ref%
{definition of relative forward} to define the best-response forward
performance criteria, denoted with a slight abuse of notation by $\left(
V_{1}(\tilde{x}_{1},t;\beta )\right) _{t\geq 0}$ and $\left( V_{2}\left( 
\tilde{x}_{2},t;\alpha \right) \right) _{t\geq 0},$ $\tilde{x}_{1}=\frac{%
x_{1}}{^{x_{2}^{\theta _{1}}}}$, $\tilde{x}_{2}=\frac{x_{2}}{x_{1}^{\theta
_{2}}}$, where $\alpha ,\beta $ stand for arbitrary policies of the
competitors. Because of symmetry, we only analyze the quantities pertinent
to manager $1.$ We provide complete characterization of her relative forward
criterion, the optimal investment and the optimal wealth processes under
relative performance concerns.

We first recall two auxiliary functions, $u_{1}:\mathbb{D}\rightarrow 
\mathbb{R}_{+}$ and $h_{1}:\mathbb{R}\times \mathbb{R}_{+}\rightarrow 
\mathbb{R}_{+}$. Function $u_{1}$ solves 
\begin{equation}
u_{1,t}=\frac{1}{2}\frac{u_{1,z}^{2}}{u_{1,zz}},\text{ }  \label{u-equation}
\end{equation}%
with initial condition given by 
\begin{equation}
\left( u_{1}^{\prime }\left( z,0\right) \right) ^{\left( -1\right)
}:=\int_{0^{+}}^{\infty }z^{-y}d\nu _{1}\left( y\right),  \label{u-initial}
\end{equation}%
for a finite positive Borel measure $\nu _{1}$.

The function $h_{1}$ is defined as $h_{1}(z,t):=\left( u_{1,z}\right)
^{\left( -1\right) }(e^{-z+\frac{t}{2}},t)$ (spatial inverse). It solves $%
h_{1,t}+\frac{1}{2}h_{1,zz}=0$ \ with \ $h_{1}\left( z,0\right)
=\int_{0^{+}}^{\infty }e^{yz}d\nu _{1}\left( y\right) $, and is given by 
\begin{equation}
h_{1}\left( z,t\right) =\int_{0^{+}}^{\infty }e^{yz-\frac{1}{2}y^{2}t}d\nu
_{1}\left( y\right) .  \label{h-1-explicit}
\end{equation}%
Let also $R_{1}:\mathbb{D}\rightarrow \mathbb{R}_{+}$, 
\begin{equation}
R_{1}(z,t):=-\frac{u_{1,z}\left( z,t\right) }{zu_{1,zz}\left( z,t\right) }=%
\frac{h_{1,z}\left( h_{1}^{\left( -1\right) }\left( z,t\right) ,t\right) }{%
h_{1}\left( h_{1}^{\left( -1\right) }\left( z,t\right) ,t\right) },
\label{risk-tolerance}
\end{equation}%
with the latter equality following from the definition of $h_{1}.$

The functions $u_{1},h_{1}$ and $R_{1}$ were introduced in \cite{MZ10a},
and used to construct in full generality the locally riskless forward
criteria in the absence of competition ($\theta _{1}=0)$; we refer the
reader therein for details, and especially for the assumptions on measure $%
\nu _{1}.$

Finally, we consider the processes $\left( A_{1,t}\right) _{t\geq 0}$ and\ $%
\left( M_{1,t}\right) _{t\geq 0}$ defined, for $\tilde{\lambda}_{1,1},\tilde{%
\lambda}_{1,2}$ as in (\ref{lamda-1-^}), as 
\begin{equation}
A_{1}:=\frac{1}{1-\rho ^{2}}\int_{0}^{t}\left( \tilde{\lambda}%
_{1,1}^{2}-2\rho \tilde{\lambda}_{1,1}\tilde{\lambda}_{1,2}+\tilde{\lambda}%
_{1,2}^{2}\right) ds\text{ \  \ and \ }M_{1}:=\int_{0}^{t}\tilde{\lambda}%
_{1,1}dW_{1}+\int_{0}^{t}\tilde{\lambda}_{1,2}dW_{2}\text{.\  \ }
\label{L-M-1}
\end{equation}

Next, we present the main result in the asset diversification case.

\begin{proposition}
\label{Competitive forward} Let policy $\beta =\left( \beta _{1},\beta
_{2}\right) \in \mathcal{A}$ and $C_{2}\left( \beta \right) $ as in (\ref%
{C-2}), and define $\left( B_{1,t}\right) _{t\geq 0}$ as 
\begin{equation}
B_{1}:=e^{\frac{1}{2}\theta _{1}\left( 1-\theta _{1}\right)
\int_{0}^{t}C_{2}\left( \beta \right) ds}.  \label{B-1-process}
\end{equation}%
Let the processes $A_{1}$ and $M_{1}$ be as in (\ref{L-M-1}), $u_{1}(\tilde{x%
}_{1},0)$ as in (\ref{u-initial}) with $u_{1}\left( z,t\right) $ solving (%
\ref{u-equation}), and introduce the process $\left( H_{1,t}\right) _{t\geq
0},$ 
\begin{align}\label{H-1}
H_{1}&:=\frac{h_{1,z}\left( h_{1}^{\left( -1\right) }\left( \tilde{x}%
_{1},0\right) +A_{1}+M_{1},A_{1}\right) }{h_{1}\left( h_{1}^{\left(
-1\right) }\left( \tilde{x}_{1},0\right) +A_{1}+M_{1},A_{1}\right) }\\
\nonumber &=\frac{\int_{0^{+}}^{\infty }e^{yh_{1}^{\left( -1\right) }\left( \tilde{x}%
_{1},0\right) \text{ }}yd\tilde{\nu}_{1,t}\left( y\right) }{%
\int_{0^{+}}^{\infty }e^{yh_{1}^{\left( -1\right) }\left( \tilde{x}%
_{1},0\right) \text{ }}d\tilde{\nu}_{1,t}\left( y\right) }, 
\end{align}
with 
\begin{equation}
d\tilde{\nu}_{1,t}\left( y\right) =e^{y(1-\frac{y}{2})A_{1}+yM_{1}}d\nu
_{1,t}\left( y\right) .  \label{random-measure}
\end{equation}%
The following assertions hold:

i)\ The process $\left( V_{1}(\tilde{x}_{1},t;\beta )\right) _{t\geq 0},$
given by 
\begin{equation}
V_{1}(\tilde{x}_{1},t;\beta )=u_{1}\left( \frac{\tilde{x}_{1}}{B_{1}}%
,A_{1}\right) ,  \label{forward-1-div}
\end{equation}%
with $V_{1}(\tilde{x}_{1},0;\beta )=u_{1}(\tilde{x}_{1},0)$ is the unique
locally riskless best-response forward criterion with such initial
condition. For each $\beta \in \mathcal{A}$ and $\tilde{x}_{1}>0,$ $V_{1}(%
\tilde{x}_{1},t;\beta )$ is time-decreasing.

ii) The optimal wealth process $\left( \tilde{X}_{1,t}^{\ast }\right)
_{t\geq 0}$ is given by 
\begin{align}\label{optimal-wealth-1-div}
\tilde{X}_{1}^{\ast }&=B_{1}h_{1}\left( h_{1}^{\left( -1\right) }\left( 
\tilde{x}_{1},0\right) +A_{1}+M_{1},A_{1}\right)\\
\nonumber  &=B_{1}\int_{0^{+}}^{\infty }e^{yh_{1}^{\left( -1\right) }\left( \tilde{x}%
_{1},0\right) \text{ }}d\tilde{\nu}_{1,t}\left( y\right) ,
\end{align}
with $\tilde{\nu}_{1,t}$ as in (\ref{random-measure}).
iii) Let\ $\alpha ^{\ast }\left( z,t\right) =\left( \alpha _{1}^{\ast
}\left( z,t\right) ,\alpha _{2}^{\ast }\left( z,t\right) \right) ,$ $\left(
z,t\right) \in \mathbb{D},$ be defined as 
\begin{equation}
\alpha _{1}^{\ast }\left( z,t\right) =\frac{\tilde{\lambda}_{1,1}-\rho 
\tilde{\lambda}_{1,2}}{\left( 1-\rho ^{2}\right) \sigma _{1}}%
B_{1}R_{1}\left( \frac{z}{B_{1}},A_{1}\right) +\theta _{1}\beta _{1}
\label{a*-feedback}
\end{equation}%
and%
\begin{equation*}
\alpha _{2}^{\ast }\left( z,t\right) =\frac{-\rho \tilde{\lambda}_{1,1}+%
\tilde{\lambda}_{1,2}}{\left( 1-\rho ^{2}\right) \sigma _{2}}%
B_{1}R_{1}\left( \frac{z}{B_{1}},A_{1}\right) +\theta _{1}\beta _{2}.
\end{equation*}%
Then, the optimal investment processes $\left( \alpha _{1,t}^{\ast }\right)
_{t\geq 0},\left( \alpha _{2,t}^{\ast }\right) _{t\geq 0}$ are given in the
feedback form, 
\begin{equation}
\alpha _{1}^{\ast }=\alpha _{1}^{\ast }\left( \tilde{X}_{1}^{\ast
},A_{1}\right) \text{ \  \  \ and \  \ }\alpha _{2}^{\ast }=\alpha _{2}^{\ast
}\left( \tilde{X}_{1}^{\ast },A_{1}\right) ,\text{\ }  \label{feedback-div}
\end{equation}%
and in closed form, 
\begin{equation}
\alpha _{1}^{\ast }=\frac{\tilde{\lambda}_{1,1}-\rho \tilde{\lambda}_{1,2}}{%
\left( 1-\rho ^{2}\right) \sigma _{1}}H_{1}+\theta _{1}\beta _{1}
\label{a-optimal-div}
\end{equation}%
and 
\begin{equation*}
\alpha _{2}^{\ast }=\frac{-\rho \tilde{\lambda}_{1,1}+\tilde{\lambda}_{1,2}}{%
\left( 1-\rho ^{2}\right) \sigma _{2}}H_{1}+\theta _{1}\beta _{2},
\end{equation*}%
with $H_{1}$ as in (\ref{H-1}).
\end{proposition}

\begin{proof}
\bigskip Let $\hat{\alpha}_{1}:=\alpha _{1}-\theta _{1}\beta _{1},$ $\hat{%
\alpha}_{2}:=\alpha _{2}-\theta _{1}\beta _{2}.$ Then, the state dynamics (%
\ref{X-1-div})\ can be written as 
\begin{equation*}
\frac{d\tilde{X}_{1}}{\tilde{X}_{1}}=\hat{\alpha}_{1}\sigma _{1}\left( 
\tilde{\lambda}_{1,1}dt+dW_{1}\right) +\hat{\alpha}_{2}\sigma _{2}\left( 
\tilde{\lambda}_{1,2}dt+dW_{2}\right) +\frac{1}{2}\theta _{1}(1-\theta
_{1})C_{2}\left( \beta \right) dt.
\end{equation*}%
Defining the auxiliary process $( \hat{X}_{1,t}) _{t\geq 0}$ by 
\begin{equation*}
\hat{X}_{1}=\frac{\tilde{X}_{1}}{B_{1}}, 
\end{equation*}%
we have that%
\begin{equation}
\frac{d\hat{X}_{1}}{\hat{X}_{1}}=\hat{\alpha}_{1}\sigma _{1}\left( \tilde{%
\lambda}_{1,1}dt+dW_{1}\right) +\hat{\alpha}_{2}\sigma _{2}\left( \tilde{%
\lambda}_{1,2}dt+dW_{2}\right) ,  \label{X-1-aux}
\end{equation}%
with $\hat{X}_{1,0}=\tilde{X}_{1,0}=\tilde{x}_{1}.$ We are, then, in the
complete market framework of \cite[Section 3]{MZ10a} and we deduce that if 
$u_{1}:\mathbb{D\rightarrow R}_{+}$ solves (\ref{u-equation}) and satisfies (%
\ref{u-initial}), then the process $u_{1}\left( \hat{X}_{1},A_{1}\right) $
is a supermartingale for any $\left( \left( \hat{\alpha}_{1,t}\right)
_{t\geq 0},\left( \hat{\alpha}_{2,t}\right) _{t\geq 0}\right) $ and becomes
a martingale for $\left( \hat{\alpha}_{1,t}^{\ast }\right) _{t\geq 0},\left( 
\hat{\alpha}_{2,t}^{\ast }\right) _{t\geq 0}$ given by 
\begin{equation*}
\hat{\alpha}_{1}^{\ast }=-\frac{\tilde{\lambda}_{1,1}-\rho \tilde{\lambda}%
_{1,2}}{\left( 1-\rho ^{2}\right) \sigma _{1}}\hat{R}^{\ast }\text{ \ and \ }%
\hat{\alpha}_{2}^{\ast }=\frac{-\rho \tilde{\lambda}_{1,1}+\tilde{\lambda}%
_{1,2}}{\left( 1-\rho ^{2}\right) \sigma _{2}}\hat{R}^{\ast },
\end{equation*}%
with $\left( \hat{R}_{t}^{\ast }\right) _{t\geq 0}=-\frac{u_{1,z}\left( \hat{%
X}_{1}^{\ast },A_{1}\right) }{\hat{X}_{1}^{\ast }u_{1,zz}\left( \hat{X}%
_{1}^{\ast },A_{1}\right) },$ where $\hat{X}_{1}^{\ast }$ solves (\ref%
{X-1-aux}) with $\left( \hat{\alpha}_{1}^{\ast },\hat{\alpha}_{2}^{\ast
}\right) $ being used. Following the analysis in \cite{MZ10a}, we deduce
that the optimal process $\left( \hat{X}_{1,t}^{\ast }\right) _{t\geq 0}$ is
given in closed form by $\hat{X}_{1}^{\ast }=h_{1}\left( h_{1}^{\left(
-1\right) }\left( \hat{x}_{1},0\right) +A_{1}+M_{1},A_{1}\right) $ and (\ref%
{optimal-wealth-1-div}) follows. Furthermore, from the definition of $h_{1}$
we deduce that 
\begin{equation*}
\hat{\alpha}_{1}^{\ast }=\frac{\tilde{\lambda}_{1,1}-\rho \tilde{\lambda}%
_{1,2}}{\left( 1-\rho ^{2}\right) \sigma _{1}}\frac{h_{1,z}\left(
h_{1}^{\left( -1\right) }\left( \hat{x}_{1},0\right)
+A_{1}+M_{1},A_{1}\right) }{h_{1}\left( h_{1}^{\left( -1\right) }\left( \hat{%
x}_{1},0\right) +A_{1}+M_{1},A_{1}\right) }\text{\  \ }
\end{equation*}%
and, similarly, 
\begin{equation*}
\text{\ }\hat{\alpha}_{2}^{\ast }=\frac{\tilde{\lambda}_{1,1}-\rho \tilde{%
\lambda}_{1,2}}{\left( 1-\rho ^{2}\right) \sigma _{1}}\frac{h_{1,z}\left(
h_{1}^{\left( -1\right) }\left( \hat{x}_{1},0\right)
+A_{1}+M_{1},A_{1}\right) }{h_{1}\left( h_{1}^{\left( -1\right) }\left( \hat{%
x}_{1},0\right) +A_{1}+M_{1},A_{1}\right) }.
\end{equation*}%
We easily deduce that $\hat{\alpha}_{1}^{\ast },\hat{\alpha}_{2}^{\ast }\in 
\mathcal{A}$ as well as the rest of the assertions for the optimal wealth
and optimal policies.

To establish the time monotonicity of $V_{1}(\tilde{x}_{1},t;\beta )$,
observe that, for each $\beta \in \mathcal{A}$ and $\tilde{x}_{1}>0,$ 
\begin{equation*}
\frac{d}{dt}V_{1}(\tilde{x}_{1},t;\beta )=-\frac{1}{2}\theta _{1}\left(
1-\theta _{1}\right) \frac{C_{2}\left( \beta \right) }{B_{1}\left( \beta
\right) }u_{1,x}\left( \frac{\tilde{x}_{1}}{\tilde{B}_{1}},A_{1}\right)+\frac{\tilde{\lambda}_{1,1}^{2}-2\rho \tilde{\lambda}_{1,1}\tilde{\lambda}%
_{1,2}+\tilde{\lambda}_{1,2}^{2}}{1-\rho ^{2}}u_{1,t}\left( \frac{\tilde{x}%
_{1}}{B_{1}},A_{1}\right) <0
\end{equation*}%
as $\theta _{1}<1,$ $C_{2}>0,$ $u_{1,x}>0$ and $u_{1,t}<0.$
\end{proof}

\begin{remark}
We note that the above best-response forward performance differs from the one introduced in \cite{MZ-QF} given by $%
u\left( x/Y_{t},Z_{t}\right) ,$ where $\left( Y_{t}\right) _{t\geq 0}
$ is a traded benchmark and $\left( Z_{t}\right) _{t\geq 0}$ a ``market-view'' process. This process is not locally riskless and its state variable is the individual wealth, and not the relative one. 
\end{remark}

Similar results may be derived for manager $2.$ Let manager $1$ follow an
arbitrary policy, say $\alpha =\left( \alpha _{1},\alpha _{2}\right) \in 
\mathcal{A}.$ If we choose $V_{2}(\tilde{x}_{2},0;\alpha
)=\int_{0^{+}}^{\infty }\tilde{x}_{2}^{-y}d\nu _{2}\left( y\right) ,$ for a
suitable positive Borel measure $\nu _{2},$ we deduce that the unique
locally riskless best-response forward criterion is given by 
\begin{equation*}
V_{2}(\tilde{x}_{2},t;\alpha )=u_{2}\left( \frac{\tilde{x}_{2}}{B_{2}}%
,A_{2}\right) ,
\end{equation*}%
with $u_{2}$ solving (\ref{u-equation}) with $u_{2}\left( z,0\right) =V_{2}(%
\tilde{x}_{2},0;\alpha ),$ $\left( B_{2,t}\right) _{t\geq 0}=$ $e^{\frac{1}{2%
}\theta _{2}\left( 1-\theta _{2}\right) \int_{0}^{t}C_{1}\left( \alpha
\right) ds},$ $\left( A_{2,t}\right) _{t\geq 0}:=\left( 1-\rho ^{2}\right)
\int_{0}^{t}\left( \tilde{\lambda}_{2,1}^{2}-2\rho \tilde{\lambda}_{2,1}%
\tilde{\lambda}_{2,2}+\tilde{\lambda}_{2,2}^{2}\right) ds$ with $\tilde{%
\lambda}_{2,1},\tilde{\lambda}_{2,2}$ and $C_{1}\left( \alpha \right) $ as
in (\ref{lambda-2-^}) and (\ref{C-1}).

Furthermore, if $\left( M_{2,t}\right) _{t\geq 0}:=\int_{0}^{t}\tilde{\lambda%
}_{2,1}dW_{1}+\int_{0}^{t}\tilde{\lambda}_{2,2}dW_{2},$ $%
h_{2}(z,t):=u_{2,z}^{\left( -1\right) }(e^{-z+\frac{t}{2}},t)$ and $\left(
H_{2}\right) _{t\geq 0}$ defined as 
\begin{equation}
H_{2}:=\frac{h_{2,z}\left( h_{2}^{\left( -1\right) }\left( \tilde{x}%
_{2},0\right) +A_{2}+M_{2},A_{2}\right) }{h_{2}\left( h_{2}^{\left(
-1\right) }\left( \tilde{x}_{2},0\right) +A_{2}+M_{2},A_{2}\right) },
\label{H-2}
\end{equation}%
then, the optimal wealth $\left( \tilde{X}_{2,t}^{\ast }\right) _{t\geq 0}$
is given by $\tilde{X}_{2}^{\ast }=B_{2}h_{2}\left( h_{2}^{\left( -1\right)
}\left( \tilde{x}_{2},0\right) +A_{2}+M_{2},A_{2}\right) $ and the policies $%
\left( \beta _{1,t}^{\ast }\right) _{t\geq 0},\left( \beta _{2,t}^{\ast
}\right) _{t\geq 0},$ with 
\begin{equation}
\beta _{1}^{\ast }=\frac{\tilde{\lambda}_{2,1}-\rho \tilde{\lambda}_{2,2}}{%
\left( 1-\rho ^{2}\right) \sigma _{2}}H_{2}+\theta _{2}\alpha _{1},\text{ \
\  \  \ }\beta _{2}^{\ast }=\frac{-\rho \tilde{\lambda}_{2,1}+\tilde{\lambda}%
_{2,2}}{\left( 1-\rho ^{2}\right) \sigma _{2}}H_{2}+\theta _{2}\alpha _{2}
\label{beta*-feedback}
\end{equation}%
are optimal.

Replacing $\tilde{\lambda}_{1,1},\tilde{\lambda}_{1,2},\tilde{\lambda}_{2,1}%
\tilde{\lambda}_{2,2}$ in (\ref{lamda-1-^}) and (\ref{lambda-2-^}), yields
the simplified forms (recall that $\alpha =\left( \alpha _{1},\alpha
_{2}\right) $ and $\beta =\left( \beta _{1},\beta _{2}\right) $) in the
original market dynamics, 
\begin{equation}
\left \{ 
\begin{array}{c}
a_{1}^{\ast }=\frac{\lambda _{1}-\rho \lambda _{2}}{\left( 1-\rho
^{2}\right) \sigma _{1}}H_{1}\left( \tilde{x}_{1},\beta \right) +\left(
1-H_{1}\left( \tilde{x}_{1},\beta \right) \right) \theta _{1}\beta _{1} \\ 
\\ 
a_{2}^{\ast }=\frac{-\rho \lambda _{1}+\lambda _{2}}{\left( 1-\rho
^{2}\right) \sigma _{1}}H_{1}\left( \tilde{x}_{1},\beta \right) +\left(
1-H_{1}\left( \tilde{x}_{1},\beta \right) \right) \theta _{1}\beta _{2},%
\end{array}%
\right.  \label{a*-div}
\end{equation}%
and%
\begin{equation}
\left \{ 
\begin{array}{c}
\beta _{1}^{\ast }=\frac{\lambda _{1}-\rho \lambda _{2}}{\left( 1-\rho
^{2}\right) \sigma _{1}}H_{2}\left( \tilde{x}_{2},\alpha \right) +\left(
1-H_{2}\left( \tilde{x}_{2},\alpha \right) \right) \theta _{2}\alpha _{1},
\\ 
\\ 
\beta _{2}^{\ast }=\frac{-\rho \lambda _{1}+\lambda _{2}}{\left( 1-\rho
^{2}\right) \sigma _{1}}H_{2}\left( \tilde{x}_{2},\alpha \right) +\left(
1-H_{2}\left( \tilde{x}_{2},\alpha \right) \right) \theta _{2}\alpha _{2}.%
\end{array}%
\right.  \label{beta*-div}
\end{equation}

\smallskip 

\textit{Discussion: }The best-response forward criterion (rewritten with
more explicit notation) is given by the locally riskless process $%
u_{1}\left( \frac{\tilde{x}_{1}}{B_{1}\left( \beta ,\rho ,\theta _{1}\right) 
},A_{1}\left( \lambda ,\sigma ,\beta ;\rho ,\theta _{1}\right) \right) .$
The process $B_{1}\left( \beta ,\rho ,\theta \right) >1$ depends only on the
competitor's policy $\beta ,$ the correlation $\rho $, and the competition
parameter $\theta _{1}$. It is increasing in $\theta _{1},$ when $\theta
_{1}\in \left( 0,\frac{1}{2}\right) ,$ and decreasing when $\theta _{1}\in
\left( \frac{1}{2},1\right) ,$ with maximum discounting at $\theta _{1}=%
\frac{1}{2}.$ The discounting vanishes at the limiting values $\theta
_{1}=0,1$. For $\rho ^{2}\neq 1,$ the process $C_{2}\left( \beta _{1},\beta
_{2};\rho \right) >0$ is jointly convex in $\left( \beta _{1},\beta
_{2}\right) $ and achieves a global minimum at $(0,0)$. The process $%
A_{1}\left( \lambda ,\sigma ,\beta ,\theta _{1}\right) $ is non-decreasing
in time and represents a stochastic time change. Furthermore, its time
derivative is convex in the competition parameter $\theta _{1}.$

The case when manager $2$ uses policies $\left( \beta _{1,t}^{0}\right)
_{t\geq 0},\left( \beta _{2,t}^{0}\right) _{t\geq 0}$ with $\beta _{1}^{0}=-%
\frac{\lambda _{1}-\rho \lambda _{2}}{\left( 1-\rho ^{2}\right) \sigma _{1}}%
, $ $\beta _{2}^{0}=-\frac{\lambda _{1}-\rho \lambda _{2}}{\left( 1-\rho
^{2}\right) \sigma _{2}}$ requires special attention. Therein, the modified
risk premia vanish at all times, $\tilde{\lambda}_{1,1}=\tilde{\lambda}%
_{1,2}=0,$ and thus the "personalized" fictitious market $\mathcal{M}%
_{1}^{d} $ \ becomes \textit{worthless}. In turn, $A_{1}=M_{1}=0,$ and $%
\alpha _{1}^{\ast }=\theta _{1}\beta _{1}^{0}$ and $\alpha _{1}^{\ast
}=\theta _{1}\beta _{2}^{0}$. Therefore, the optimal risky strategy is to
simply follow fraction $\theta _{1}$ of this specific competitor's strategy.
This yields $\tilde{X}_{1}^{\ast }=B_{1}\tilde{x}_{1},$ with $B_{1}:=e^{%
\frac{1}{2}\theta _{1}\left( 1-\theta _{1}\right) \left( 1-\rho ^{2}\right)
\int_{0}^{t}\left( \lambda _{1}+\lambda _{2}\right) ^{2}ds}$ and, thus, $%
\left( V_{1}\left( \tilde{X}_{1}^{\ast },t;\beta ^{0}\right) \right) _{t\geq
0}=u_{1}\left( \tilde{x}_{1},0\right) $. This is intuitively pleasing and
consistent with the fact that, in a worthless market, the performance
criterion should \textit{not} change with time (provided all quantities are
expressed in discounted units).

The optimal policy is given both via a feedback and in closed form (cf.~(\ref%
{a*-feedback})\ and (\ref{a-optimal-div})). The feedback control depends on
wealth only through the random function $R_{1}\left( z,t\right) ,$ which is
the relative risk tolerance associated with $u_{1}\left( z,t\right) .$ Using
the results in \cite{MZ10a}, we deduce that $R_{1}\left( z,t\right) ,$ and
thus $\alpha _{1}\left( z,t\right) $ and $\alpha _{2}\left( z,t\right) ,$
are decreasing in time and non-increasing in $z.$

\subsubsection{The CRRA case}

Let the measure in (\ref{u-initial}) be a Dirac, $\nu _{1}\left( dy\right)
=\delta _{\frac{1}{\gamma _{1}}},$ $\gamma _{1}>0.$ Then, $h_{1}\left(
z,t\right) =e^{\frac{z}{\gamma _{1}}-\frac{1}{2}\left( \frac{1}{\gamma _{1}}%
\right) ^{2}t}$ and $H_{1}\left( \tilde{x}_{1},\beta \right) =\frac{1}{%
\gamma _{1}}$ (cf. (\ref{h-1-explicit}) and (\ref{H-1})). Criterion (\ref%
{forward-1-div})\ becomes 
\begin{equation*}
V_{1}(\tilde{x}_{1},t;\beta )=\frac{1}{1-\gamma _{1}}\left( \frac{\tilde{x}%
_{1}}{B_{1}}\right) ^{1-\gamma _{1}}e^{-\frac{1}{2}\frac{1-\gamma _{1}}{%
\gamma _{1}}A_{1}},
\end{equation*}%
and, this is the unique locally riskless homothetic criterion associated
with $\gamma _{1}$. The optimal policies and optimal wealth processes are
given by (\ref{a-optimal-div}) and (\ref{optimal-wealth-1-div}),%
\begin{equation*}
\alpha _{1}^{\ast }=\frac{1}{\gamma _{1}}\frac{\lambda _{1}-\rho \lambda _{2}%
}{\left( 1-\rho ^{2}\right) \sigma _{1}}+\left( 1-\frac{1}{\gamma _{1}}%
\right) \theta _{1}\beta _{1},
\end{equation*}%
\begin{equation*}
\alpha _{2}^{\ast }=\frac{1}{\gamma _{1}}\frac{-\rho \lambda _{1}+\lambda
_{2}}{\left( 1-\rho ^{2}\right) \sigma _{2}}+\left( 1-\frac{1}{\gamma _{1}}%
\right) \theta _{1}\beta _{2}
\end{equation*}%
and $\tilde{X}_{1}^{\ast }=\tilde{x}_{1}e^{\frac{1}{\gamma _{1}}\left( 1-%
\frac{1}{2\gamma _{1}}\right) A_{1}+\frac{1}{\gamma _{1}}M_{1}}B_{1}.$ We
recall that there is no assumption for the preferences of manager $2$, only
that she follows an arbitrary policy $\beta \in \mathcal{A}.$

\bigskip

Similarly, let manager $1$ follow policy $\alpha =\left( \alpha _{1},\alpha
_{2}\right) .$ If $\nu _{2}\left( dy\right) =\delta _{\frac{1}{\gamma _{2}}%
}, $ for $\gamma _{2}>0,$ $\gamma _{2}\neq 1,$ the unique locally riskless
best-response forward criterion for manager $2$ with initial condition $%
V_{2}(\tilde{x}_{2},0;\alpha )=\frac{\tilde{x}_{2}^{1-\gamma _{2}}}{1-\gamma
_{2}}$ is given by%
\begin{equation*}
V_{2}(\tilde{x}_{2},t;\alpha )=\frac{1}{1-\gamma _{2}}\left( \frac{\tilde{x}%
_{2}}{B_{2}}\right) ^{1-\gamma _{2}}e^{-\frac{1}{2}\frac{1-\gamma _{2}}{%
\gamma _{2}}A_{2}},
\end{equation*}%
and the optimal policies and optimal wealth by%
\begin{equation}
\beta _{1}^{\ast }=\frac{1}{\gamma _{2}}\frac{\lambda _{1}-\rho \lambda _{2}%
}{\left( 1-\rho ^{2}\right) \sigma _{1}}+\left( 1-\frac{1}{\gamma _{2}}%
\right) \theta _{2}a_{1}\text{, }\quad\beta _{2}^{\ast }=\frac{1}{\gamma _{2}}%
\frac{-\rho \lambda _{1}+\lambda _{2}}{\left( 1-\rho ^{2}\right) \sigma _{2}}%
+\left( 1-\frac{1}{\gamma _{2}}\right) \theta _{2}a_{2}
\label{beta-power-div}
\end{equation}%
and $\tilde{X}_{2,t}^{\ast }=\tilde{x}_{2}e^{\frac{1}{\gamma _{2}}\left( 1-%
\frac{1}{2\gamma _{2}}\right) A_{2}+\frac{1}{\gamma _{2}}M_{2}}B_{2}.$

\subsection{Forward Nash equilibrium}

\label{forward nash equilibrium}

The forward Nash equilibrium is defined as in Definition \ref{definition of
forward Nash}. To find the equilibrium strategies $\left( \alpha _{t}^{\ast
}\right) _{t\geq 0}$, $\left( \beta _{t}^{\ast }\right) _{t\geq 0}$ one then
needs to solve the non-linear system (cf. (\ref{a*-div}) and (\ref{beta*-div}%
)),%
\begin{equation}
\left \{ 
\begin{array}{c}
a_{1}^{\ast }=\frac{\lambda _{1}-\rho \lambda _{2}}{\left( 1-\rho
^{2}\right) \sigma _{1}}H_{1}\left( \tilde{x}_{1},\beta ^{\ast }\right)
+\left( 1-H_{1}\left( \tilde{x}_{1},\beta ^{\ast }\right) \right) \theta
_{1}\beta _{1}^{\ast } \\ 
a_{2}^{\ast }=\frac{-\rho \lambda _{1}+\lambda _{2}}{\left( 1-\rho
^{2}\right) \sigma _{2}}H_{1}\left( \tilde{x}_{1},\beta ^{\ast }\right)
+\left( 1-H_{1}\left( \tilde{x}_{1},\beta ^{\ast }\right) \right) \theta
_{1}\beta _{2}^{\ast } \\ 
\beta _{1}^{\ast }=\frac{\lambda _{1}-\rho \lambda _{2}}{\left( 1-\rho
^{2}\right) \sigma _{1}}H_{2}\left( \tilde{x}_{2},\alpha ^{\ast }\right)
+\left( 1-H_{2}\left( \tilde{x}_{2},\alpha ^{\ast }\right) \right) \theta
_{2}\alpha _{1}^{\ast } \\ 
\beta _{2}^{\ast }=\frac{-\rho \lambda _{1}+\lambda _{2}}{\left( 1-\rho
^{2}\right) \sigma _{2}}H_{2}\left( \tilde{x}_{2},\alpha ^{\ast }\right)
+\left( 1-H_{2}\left( \tilde{x}_{2},\alpha ^{\ast }\right) \right) \theta
_{2}\alpha _{2}^{\ast }.%
\end{array}%
\right.  \label{system-div}
\end{equation}%
The system is in general difficult to solve unless for special cases, one of
which is examined next.

\subsubsection{The CRRA case}

We derive explicit solutions when both managers have homothetic forward
criteria using (\ref{beta-power-div}) and (\ref%
{system-div}).

\begin{proposition}
\label{pro:Nash-CARRA-dive} Let $\gamma _{1},\gamma _{2}>0$ with $\gamma
_{1},\gamma _{2}\neq 1,$ and assume that $\gamma _{1}\gamma _{2}-\theta _{1}\theta _{2}(1-\gamma
_{1})(1-\gamma _{2})\neq 0$. Then, the Nash equilibrium strategies $\left(
\alpha _{t}^{\ast }\right) _{t\geq 0}$, $\left( \beta _{t}^{\ast }\right)
_{t\geq 0}$ are given as  
\begin{align}
\label{pi-1-Nash-dive} \alpha^{\ast}_{1} &= c_{\alpha}\frac{\lambda_{1}-\rho\lambda_{2}}{\sigma_{1}(1-\rho^2)}\quad\text{and}\quad \alpha^{\ast}_{2}=c_{\alpha}\frac{\lambda_{2}-\rho\lambda_{1}}{\sigma_{2}(1-\rho^2)}\\
\label{pi-2-Nash-dive} \beta^{\ast}_{1} &= c_{\beta}\frac{\lambda_{1}-\rho\lambda_{2}}{\sigma_{1}(1-\rho^2)}\quad\text{and}\quad \beta^{\ast}_{2}=c_{\beta}\frac{\lambda_{2}-\rho\lambda_{1}}{\sigma_{2}(1-\rho^2)}
\end{align}
where the constants $c_{\alpha}$ and $c_{\beta}$ are defined as
\begin{equation*}
c_{\alpha}:=\frac{\gamma _{2}+\theta _{1}(\gamma _{1}-1)}{\gamma _{1}\gamma _{2}-\theta _{1}\theta _{2}(1-\gamma
_{1})(1-\gamma _{2}) },\quad c_{\beta}:=\frac{\gamma _{1}+\theta _{2}(\gamma _{2}-1)}{\gamma _{1}\gamma _{2}-\theta _{1}\theta _{2}(1-\gamma
_{1})(1-\gamma _{2}) }.
\end{equation*}
\end{proposition}
\begin{proof}
Taking into account \eqref{beta-power-div}, we get that system \eqref{system-div} becomes linear. Assumption $\gamma _{1}\gamma _{2}-\theta _{1}\theta _{2}(1-\gamma
_{1})(1-\gamma _{2})\neq 0$ guarantees that the determinant is different than zero and hence the system admits a unique solution. Simple calculations imply \eqref{pi-1-Nash-dive} and \eqref{pi-2-Nash-dive} and standing assumptions on $\lambda_1$ and $\lambda_2$ yield the admissibility of the equilibrium investment strategies.   
\end{proof}

Similarly to the asset specialization setting, the Nash equilibrium strategies \eqref{pi-1-Nash-dive} and \eqref{pi-2-Nash-dive} have the same form as the ones in the
log-normal market and backward utility maximization criteria (see \cite[Proposition 2]{Basak15}). All conclusions in \cite{Basak15} hold for the general It\^{o}-diffusion setting we assume herein.

\section{Conclusions and extensions}

We have studied portfolio allocations of two fund managers when they
incorporate relative performance concerns. We have looked at the asset
specialization and asset diversification settings in an It\^{o}-diffusion market.
For both these cases, we have considered the best response and the Nash equilibria. We studied these issues in a new framework we introduce
herein that is based on forward performance criteria. These criteria allow
for ``real-time" updating of both the model coefficients and the competitor's
policies as well as for flexible horizons. Thus, we considerably generalize
the existing work on the subject by allowing i) a considerably more general market model, ii)\ no a priori modeling of the competitor's policy
and iii)\ flexible investment horizons. Next, we discuss some possible
extensions.

\textit{i) Multi-frequencies:\ }In all cases herein, we have assumed that
model selection, trading and relative/competitive performance valuation are
all aligned and, furthermore, that they all occur continuously in time. In
reality, however, these three fundamental attributes are not synchronized. A
more realistic scenario would allow trading to take place more frequently
than model selection, and relative performance evaluation to occur less
frequently than trading. Note that the most extreme case is in the classical
expected utility problem in which the terminal utility is specified only
once, at initial time, with no further risk preference adjustment.

With regards to the relative frequency of trading and model selection, it is
more realistic to assume that the model is selected for some trading period
ahead, say a week, and that within this week, trading takes place in
discrete or continuous time. When relative performance is involved, the
distinct scales of time evolution are more critical, for each fund manager
typically announces her performance at discrete times and not continuously.

\textit{ii) Information about market and competitors: }Information
availability and acquisition for both the market and the competitor's
behavior and performance are of tantamount importance. In the existing
literature it is assumed that both managers have full access to both the
market(s) and risk preferences. While we relax the requirement that neither
the model dynamics nor the competitor's input (risk preferences, chosen
policy and investment horizon) need to be a priori modeled, we do assume
that any information - acquired in real time - about them is available to
both managers, together with their relative bias parameters. These
assumptions are partially supported by existing results; see, for example, 
\cite{Koijen14}, where it is argued that managers acquire such information
from the realized, and publicly available, returns of their piers.

However, several \textquotedblleft under-specification" issues remain open,
especially in terms of the manager's risk preferences, specialized knowledge
and past performance. For example, it might be more realistic to assume that
at the end of each relative evaluation period, each fund manager receives
information about the performance of the other and, right after, formulates
a view about the possible upcoming performance till the end of the next
evaluation period. This will partially address the absence of complete
information under asset specialization. In this case, injecting personal
views could lead to a forward Black-Litterman type criterion under
competition.

\textit{iii) Beyond locally riskless and reduced form relative
performance/competition criteria: }Herein, we worked with criteria that are,
from the one hand, locally riskless processes and, from the other, of the
``homogeneous" scaling (\ref{metric1}) and (\ref{metric2}). In general,
relative performance concerns might be modeled, at the level of the
criterion, by arbitrary $\mathcal{F}_{t}$-adapted processes, say $%
C_{1}(x_{2},t)$ and $C_{2}(x_{1},t).$ These processes might then model, in a
more refined way, the competition dependence on past performance of the
competitors, market conditions and time in a more realistic way.

In a different direction, the forward criteria might have volatility, which
would capture uncertainty about the model dynamics and/or the competitor's
beliefs and policies. We will then work with criteria of the form 
\begin{equation*}
dU_{1}\left( x_{1},x_{2},t\right) =b_{1,t}\left(
x_{1},C_{1}(x_{2},t),t\right) dt+a_{1,1,t}\left( x_{1},C_{1}(x_{2},t),t\right) dW_{1,t}+\alpha
_{1,2,t}\left( x_{1},C_{1}(x_{2},t),t\right) dW_{2,t}
\end{equation*}%
and 
\begin{equation*}
dU_{2}\left( x_{1},x_{2},t\right) =b_{2,t}\left(
C_{2}(x_{1},t),x_{2},t\right) dt+a_{2,1,t}\left( C_{2}(x_{1},t),x_{2}\right) dW_{1,t}+a_{2,2,t}\left(
C_{2}(x_{1},t),x_{2},t\right) dW_{2,t},
\end{equation*}%
with the volatilities $\left( a_{1,1,t},a_{1,2,t}\right)_{t\geq 0} $ and $\left(
a_{2,1,t},a_{2,2,t}\right)_{t\geq 0} $ being adapted and manager-specific input
processes. Proceeding as in \cite{MZ3} we would
then obtain a stochastic PDE (rather than a random one) with coefficients
depending on the evolving market dynamics and the competitor's policies. As
in the absence of relative concerns, these equations will be ill-posed and
degenerate with little, if any, tractability. In turn, the systems related
to the forward Nash equilibria (cf.~(\ref{system}) and (\ref{system-div}))
would be systems of such infinite dimensional equations.

In a different direction, relative forward criteria may be modeled as
discrete or a combination of discrete and continuous-time processes for
different, possible nested, time regimes, associated with distinct
frequencies as discussed above. For discrete processes, predictability is a
natural assumption (see \cite{Ango16} for a binomial model and adaptive
market parameter selection).

\bibliographystyle{abbrv}
\bibliography{mybib}

\end{document}